\documentclass[11pt,letterpaper]{article}

\usepackage[margin=1in]{geometry}
\usepackage[utf8]{inputenc}

\usepackage{amsmath,amssymb,amsthm,mathrsfs}
\usepackage{physics}
\usepackage{complexity}
\usepackage{caption,subcaption}
\usepackage{algorithm}
\usepackage{algpseudocode}
\usepackage{setspace}

\usepackage{graphicx}
\usepackage[colorlinks=true]{hyperref}
\usepackage[dvipsnames]{xcolor}
\definecolor{darkred}  {rgb}{0.5,0,0}
\definecolor{darkblue} {rgb}{0,0,0.5}
\definecolor{darkgreen}{rgb}{0,0.5,0}
\hypersetup{
urlcolor      = blue,         
linkcolor     = darkblue,     
citecolor     = darkgreen,    
filecolor     = darkred,
linkcolor     = darkblue
}
\usepackage{aliascnt}

\usepackage{tcolorbox}
\tcbuselibrary{breakable}
\usepackage{caption,subcaption}
\usepackage{tikz-cd}
\usetikzlibrary{decorations.pathreplacing, positioning,calc,arrows.meta,shapes.geometric}
\pgfmathsetmacro\MathAxis{height("$\vcenter{}$")}

\usepackage[toc,page]{appendix}
\usepackage{tikz,pgfplots}\pgfplotsset{compat=1.16}
\usetikzlibrary{positioning,fit,calc}

\definecolor{darkred}{rgb}{0.5,0,0}
\definecolor{darkblue}{rgb}{0,0,0.5}
\definecolor{darkgreen}{rgb}{0,0.5,0}
\definecolor{DarkRed}{RGB}{170,0,0}

\makeatletter
\@ifpackageloaded{hyperref}{}{
  \usepackage[colorlinks=true]{hyperref}
}
\makeatother

\hypersetup{
  urlcolor  = blue,
  linkcolor = darkblue,
  citecolor = darkgreen,
  filecolor = darkred
}

\theoremstyle{plain}

\newtheorem{theorem}{Theorem}[section]
\newtheorem*{theorem*}{Theorem}

\newaliascnt{proposition}{theorem}

\aliascntresetthe{proposition}

\newaliascnt{conjecture}{theorem}

\aliascntresetthe{conjecture}

\newaliascnt{lemma}{theorem}
\newtheorem{lemma}[lemma]{Lemma}
\aliascntresetthe{lemma}

\newaliascnt{claim}{theorem}

\aliascntresetthe{claim}

\newaliascnt{fact}{theorem}
\newtheorem{fact}[fact]{Fact}
\aliascntresetthe{fact}

\newaliascnt{corollary}{theorem}
\newtheorem{corollary}[corollary]{Corollary}
\aliascntresetthe{corollary}

\theoremstyle{definition}

\newaliascnt{definition}{theorem}
\newtheorem{definition}[definition]{Definition}
\aliascntresetthe{definition}

\newaliascnt{remark}{theorem}
\newtheorem{remark}[remark]{Remark}
\aliascntresetthe{remark}

\newaliascnt{properties}{theorem}

\aliascntresetthe{properties}

\newaliascnt{assumption}{theorem}
\newtheorem{assumption}[assumption]{Assumption}
\aliascntresetthe{assumption}

\newaliascnt{observation}{theorem}
\newtheorem{observation}[observation]{Observation}
\aliascntresetthe{observation}

\newaliascnt{example}{theorem}

\aliascntresetthe{example}

\newaliascnt{question}{theorem}
\newtheorem{question}[question]{Question}
\aliascntresetthe{question}

\usepackage[nameinlink, capitalize]{cleveref}

\crefname{theorem}{Theorem}{Theorems}
\Crefname{theorem}{Theorem}{Theorems}
\crefname{proposition}{Proposition}{Propositions}
\Crefname{proposition}{Proposition}{Propositions}
\crefname{conjecture}{Conjecture}{Conjectures}
\Crefname{conjecture}{Conjecture}{Conjectures}
\crefname{lemma}{Lemma}{Lemmas}
\Crefname{lemma}{Lemma}{Lemmas}
\crefname{claim}{Claim}{Claims}
\Crefname{claim}{Claim}{Claims}
\crefname{fact}{Fact}{Facts}
\Crefname{fact}{Fact}{Facts}
\crefname{corollary}{Corollary}{Corollaries}
\Crefname{corollary}{Corollary}{Corollaries}
\crefname{definition}{Definition}{Definitions}
\Crefname{definition}{Definition}{Definitions}
\crefname{remark}{Remark}{Remarks}
\Crefname{remark}{Remark}{Remarks}
\crefname{properties}{Properties}{Properties}
\Crefname{properties}{Properties}{Properties}
\crefname{assumption}{Assumption}{Assumptions}
\Crefname{assumption}{Assumption}{Assumptions}
\crefname{observation}{Observation}{Observations}
\Crefname{observation}{Observation}{Observations}
\crefname{example}{Example}{Examples}
\Crefname{example}{Example}{Examples}
\crefname{protocol}{Protocol}{Protocols}
\Crefname{protocol}{Protocol}{Protocols}
\crefname{equation}{Equation}{Equations}
\Crefname{equation}{Equation}{Equations}
\crefname{section}{Section}{Sections}
\Crefname{section}{Section}{Sections}
\crefname{question}{Question}{Question}

\usepackage{enumerate}

\renewcommand{\epsilon}{\varepsilon}

\providecommand{\tr}{\operatorname{tr}}

\newcommand{\labs}[1]{\left\lvert #1\right\rvert}

\providecommand{\MathAxis}{2.5}

\providecommand{\BC}{\mathbb{C}}

\providecommand{\BR}{\mathbb{R}}

\providecommand{\CB}{\mathcal{B}}
\providecommand{\CC}{\mathcal{C}}

\providecommand{\CE}{\mathcal{E}}
\providecommand{\CF}{\mathcal{F}}

\providecommand{\calH}{\mathcal{H}}

\providecommand{\CK}{\mathcal{K}}
\providecommand{\CL}{\mathcal{L}}

\providecommand{\CN}{\mathcal{N}}
\providecommand{\CO}{\mathcal{O}}
\providecommand{\calP}{\mathcal{P}}

\providecommand{\CS}{\mathcal{S}}

\providecommand{\vA}{\mathbf{A}}
\providecommand{\vB}{\mathbf{B}}
\providecommand{\vC}{\mathbf{C}}

\providecommand{\vH}{\mathbf{H}}
\providecommand{\vI}{\mathbf{I}}

\providecommand{\vO}{\mathbf{O}}
\providecommand{\vP}{\mathbf{P}}

\providecommand{\vX}{\mathbf{X}}
\providecommand{\vY}{\mathbf{Y}}
\providecommand{\vZ}{\mathbf{Z}}

\providecommand{\vh}{\mathbf{h}}

\providecommand{\vrho}{\boldsymbol{\rho}}
\providecommand{\vsigma}{\boldsymbol{\sigma}}

\providecommand{\sA}{\mathsf{A}}
\providecommand{\sB}{\mathsf{B}}

\providecommand{\sI}{\mathsf{I}}

\providecommand{\sM}{\mathsf{M}}

\providecommand{\sd}{\mathsf{d}}
\providecommand{\sk}{\mathsf{k}}
\providecommand{\sJ}{\mathsf{J}}
\providecommand{\sq}{\mathsf{q}}
\providecommand{\sF}{\mathsf{F}}
\providecommand{\sU}{\mathsf{U}}

\newcommand{\sT}{\mathsf{T}}

\newcommand{\rd}{\mathrm{d}}
\newcommand{\triple}[1]{\vert\!\vert\!\vert #1 \vert\!\vert\!\vert}

\usepackage{xcolor}
\usepackage{soul}
\setstcolor{red}

\usepackage[backend=biber,style=alphabetic,url=false,isbn=false,maxnames=10,minnames=3,maxalphanames=4,minalphanames=3]{biblatex}
\begin{document}

\title{Fast mixing of all-to-all quantum systems \\at high temperatures}
\author{Thiago Bergamaschi\thanks{Department of EECS, UC Berkeley, CA, USA. \href{mailto:thiagob@berkeley.edu}{thiagob@berkeley.edu}}}

\maketitle

\begin{abstract}

 It is shown that arbitrary quantum $k$-local Hamiltonians with bounded strength interactions admit a quantum Gibbs sampler \cite{chen2023efficient} with a system-size independent spectral gap, at sufficiently high temperatures. 
 As a consequence, such systems admit fully-polynomial time quantum approximation algorithms for partition functions and global expectation values.

\end{abstract}

\setcounter{tocdepth}{2}
\tableofcontents
\newpage

\section{Introduction}

A central theme in quantum many-body physics is that locality constrains dynamics. 
In \textit{geometrically local} systems (that is, lattices) such constraints are typically captured by Lieb-Robinson bounds \cite{Lieb1972, hastings2006spectral, nachtergaele2010LRB}.
These bounds quantify the speed at which information percolates throughout the system, confining the time-evolution of local operators, to an effective \textit{lightcone} around the operator. 
Lieb-Robinson bounds have played a foundational role in our understanding of locality  in quantum systems, underlying results on tensor network descriptions and area laws \cite{hastings2006solving, HastingsWen2005, hastingsAreaLaw, Alhambra_2021}, quantum and classical simulation algorithms \cite{Osborne_2006, haah2020quantum, wild2023classical, McDonough_2025}, and correlation-decay and stability properties \cite{hastings2006spectral, nachtergaele2006lieb, Bravyi2006, bravyi2010topological}.

At the same time, many quantum systems of interest are not naturally described by lattice geometry. Examples include electronic-structure Hamiltonians \cite{WLA13, ogorman2021electronicstructurefixedbasis}, strongly correlated systems \cite{Sachdev_1993, SYK, Maldacena_rmkSYK}, mean-field models \cite{toth90, Bj_rnberg_2013, Bj_rnberg_2016}, expander-based qLDPC codes \cite{Tillich_2014, Hastings_2021, panteleev2022asymptoticallygoodquantumlocally, dinur2025expansionhigherdimensionalcubicalcomplexes}, all of which exhibit long-range or non-geometric interactions. Understanding how such \textit{all-to-all systems} scramble quantum information under time-evolution can be significantly more challenging, since 
the usual notion of a lightcone is no longer the correct measure of operator growth  \cite{Lashkari_2013,Lucas_2020, chen2021operator}.  

In this paper, we consider such families of all-to-all interacting Hamiltonians, and study the complexity of their thermal (or `Gibbs') state-preparation algorithms. Following the modern quantum Gibbs sampling paradigm, we study their dissipative state preparation under Lindbladian evolution (or, a continuous-time quantum Markov chain) \cite{Rall2023thermalstate,chen2023quantum,chen2023efficient, ding2024efficient, jiang2024quantum}. Recently, there has been significant interest in proofs of efficient convergence, or \textit{mixing times} for these algorithms, albeit, existing arguments for non-commuting systems are largely limited to lattice or bounded degree Hamiltonians \cite{rouze2024efficient, rouze2024optimal, bakshi2025dobrushinconditionquantummarkov, BC26, vsmid2025rapid, tong2025fermions}. In part, this is since the analysis in prior work fundamentally relies on Lieb-Robinson bounds.

The main result in this paper is a proof that arbitrary all-to-all $\sk$-local Hamiltonians with bounded strength interactions (see \cref{def:all-to-all}) admit a  Gibbs sampler \cite{chen2023efficient} with a system-size independent spectral gap at sufficiently high temperatures. Thus, such systems are \textit{fast-mixing}, and consequently admit fully-polynomial time quantum approximation schemes for partition functions and expectation values. 
Our argument is based on revisiting the recently introduced non-commutative versions of the Dobrushin conditions, sufficient conditions for rapid-mixing in quantum systems \cite{bakshi2025dobrushinconditionquantummarkov,MajewskiOlkiewiczZegarlinski1998, temme2015fast, rouze2024optimal}, under the quantum cluster expansion of  Neto\v{c}n\'y
and Redig \cite{Neto2004, Mann_2021}. 
In the process, we hope that our techniques shed some light on the challenges in studying mixing times in all-to-all quantum systems, and serve as a stepping stone towards more complex disordered systems; see \cref{section:related} for further discussion.


\subsection{Contributions}

We consider generic $\sk$-local Hamiltonians on systems of $n$ qudits. We will state our results as a function of three notions of locality for such systems. 

\begin{definition}
    [All-to-all Hamiltonians]\label{def:all-to-all} Let $\vH:=\sum_{e\in \Gamma} \vh_e\otimes \vI_{[n]\setminus e}$ be a local Hamiltonian on $n$ qudits of local dimension $2^\sq$. The \textit{locality} $\sk$ of $\vH$ is the uniform support size $|\mathsf{supp}(\vh_e)| = \mathsf{k}.$ The \textit{degree} $\sd$ is the number of qudits any given qudit interacts with 
        \begin{equation}
        u\in [n]:\quad     \big|\{v:\exists e\in \Gamma \text{ s.t. } \{u, v\}\subseteq e\}\big|\leq \sd
        \end{equation}
   The \textit{interaction strength} $\sJ$ is the sum over terms coupling a pair of qudits $ u\neq v: \sum_{e\supseteq\{u,v\}} \|\vh_e\|\leq \sJ$.
\end{definition}

We focus on the regime where $\sk, \sq$ and $(\sd\cdot \sJ)$ are fixed constants independent of $n$.

\begin{remark}\label{remark:simple-case}
    Our statements are phrased as a function of $\sk, \sd, \sJ$ and $\sq$ for completeness. In a first pass, the reader may want to imagine a qubit $\sk=2$ local Hamiltonian (a graph) of bounded degree $\sd$ where each edge-term is bounded by $\sJ=\sd^{-1}$. We emphasize $\sd$ is allowed to scale with $n$.
\end{remark}

Our goal is to study the rate of (open-system) thermalization \cite{Breuer2006open} of such Hamiltonians, when weakly-coupled to a heat-bath at some fixed (constant) inverse-temperature $\beta$, eventually converging to its Gibbs state:
\begin{equation}
    \vrho:= \frac{e^{-\beta \vH}}{\tr[e^{-\beta \vH}]}.
\end{equation}

Following the modern paradigm in the quantum Gibbs sampling literature, we model this thermalization process by the evolution of a quantum Markov semigroup, generated by the thermal Lindbladian $\CL$ of \cite{chen2023efficient}:
\begin{equation}
    \frac{\rd}{\rd t}\vsigma=\CL[\vsigma], \quad \text{where} \quad \CL[\vsigma]=\sum_{a\in \CS_{[n]}^1}\CL_a [\vsigma], \quad \text{ and }\quad \CL_a[\vrho]=0,\label{eq:lind-eq-intro}
\end{equation}
which satisfies some form of quasi-locality and (KMS) detailed-balance, thus fixing the Gibbs state.\footnote{The detailed structure of the dynamics depends on $\vH, \beta$ and a choice of updates (here, single-site Pauli operators $\CS_{[n]}^1$); see \cref{section:prelim} for further background.} The main result of this paper is a proof of a spectral gap for the \cite{chen2023efficient} Lindbladian $\CL$ for arbitrary all-to-all Hamiltonians, at sufficiently high temperatures.

\begin{theorem}
    [Fast-mixing of all-to-all quantum systems at high temperatures]\label{theorem:all-to-all-mixing} 
    Let $\vH$ be an all-to-all Hamiltonian as in \cref{def:all-to-all}. Then, there exists a universal constant $\mathsf{c}\geq 1$ and threshold $\beta_\mathsf{c}:=[\mathsf{c}^{\sq\cdot \sk}\cdot \sd\sJ]^{-1} $ such that for all $\beta<\beta_\mathsf{c}$ the Lindbladian $\CL$ \eqref{eq:lind-eq-intro} admits a system-size independent spectral gap, depending only on $\sq, \sk$.
\end{theorem}

 In the context of balanced 2-local Hamiltonians as in \cref{remark:simple-case}, $\beta_\mathsf{c}$ is simply some explicit universal constant, which assuming $\mathsf{BQP}\neq \mathsf{NP}$ is optimal due to the hardness of approximate counting \cite{sly2010computational, sly2012computational}. To the extent of our knowledge, prior proofs of unconditional, system-size independent spectral gaps for quantum Markov chains have been restricted to bounded degree Hamiltonians. Notably, one-dimensional systems at all temperatures \cite{kastoryano2016commuting, BC26}, and in finite-dimensional lattices at high-temperatures \cite{rouze2024efficient, rouze2024optimal, bakshi2025dobrushinconditionquantummarkov}.\footnote{Albeit, these systems have been shown to be rapid-mixing, i.e. $\CO(\log \frac{n}{\epsilon})$ mixing time.}

\begin{corollary}
    [Quantum algorithms for partition functions and expectation values]
    In the context of \cref{theorem:all-to-all-mixing}, for any $\epsilon>0$ there exists a $\mathsf{poly}(n,\frac{1}{\epsilon})$ time quantum algorithm to estimate the partition function $Z_\beta:=\tr[e^{-\beta \vH}]$ up to relative error $1\pm \epsilon$, and global expectation values up to additive error $\epsilon$, in all-to-all quantum systems at sufficiently high temperatures.
\end{corollary}

Indeed, \cref{theorem:all-to-all-mixing} implies that the evolution converges to within $\epsilon$ trace distance of $\vrho$ in time $t_{\mathsf{mix}}(\epsilon)\leq \CO(n+ \log \frac{1}{\epsilon})$, entailing efficient state-preparation on a quantum computer by simulation of the dynamics \cite[Theorem I.2]{chen2023efficient}. We remark that existing classical algorithms for the same task \cite{Mann_2021, harrow2020classical, yin2023, bakshi2024high, zlokapa2026rigorousquasipolynomialtimeclassicalalgorithm} require $n^{\CO(\log \sd)}$ time. In fact, it is conjectured in  \cite[Section V]{Mann_2021} that improved runtimes could be achieved via Markov chain methods.


\subsection{Discussion and related work}\label{section:related}

\textbf{Cluster expansions.} As we discuss shortly, the starting point behind \cref{theorem:all-to-all-mixing} is a cluster expansion for the short (complex) time evolution of all-to-all Hamiltonians. While Lieb-Robinson bounds for lattice Hamiltonians quantify the locality of the time-evolution of local observables as a function of \textit{distance} on the lattice \cite{Lieb1972,hastings2006spectral, nachtergaele2010LRB, chen2023speed}, such cluster expansions will express the time evolution as linear combinations of connected operators \cite{Lashkari_2013, Lucas_2020, chen2021operator}. 

Cluster expansions have a rich history in classical and quantum many-body physics, including e.g. establishing uniqueness and analyticity of the Gibbs measure \cite{dobrushin1985completely, dobrushin1985constructive, KoteckyPreiss1986, harrow2020classical}, correlation decay properties \cite{ueltschi2005clusterexpansionscorrelationfunctions, Kliesch_2014, Fr_hlich_2015},
algorithms for partition functions and local expectation values at high temperatures \cite{Mann_2021, yin2023, wild2023classical, Mann_2024,
hongrui2025convergencecumulantexpansionpolynomialtime} and low temperatures \cite{Helmuth2019, Chen2019FastAA, Ber23, HM23lowT}, concentration inequalities \cite{Neto2004, KuwaharaSaito2020Gaussian, cramer2026berryesseenboundquantumlattice}. The specific quantum cluster expansion we use is based on the approach in \cite{Neto2004} and its algorithmic version in \cite{Mann_2021}; which in turn is based on the seminal classical Koteck\'y and Preiss condition \cite{KoteckyPreiss1986}. 

\paragraph{Discussion on classical simulation.} Perhaps the most direct applications of our techniques would be to design new or improve existing classical simulation algorithms for all-to-all systems. Although we do not attempt this effort here for conciseness, we believe the complex-time operator growth bounds of \cref{section:quasi-locality} could be leveraged for several different styles of classical algorithms:
\begin{enumerate}
    \item To prove \textit{separability} (unentanglement) of such systems at high temperatures, under the ``quasi-local perturbations of identity'' approach of \cite{bakshi2024high} (see also \cite{Scalet2026SpatialES, bakshi2026entanglementquantumspinchains}). 
    \item To prove a zero-free region in $\beta$ for the complex-valued partition function, in turn implying classical quasi-polynomial  time algorithms under Barvinok's method \cite{barvbook, barvinok2014computingpartitionfunctioncliques, barvinok2014computingpermanentsomecomplex, harrow2020classical}. See also \cite{zlokapa2026sykthermalexpectationsclassically, zlokapa2026rigorousquasipolynomialtimeclassicalalgorithm} for a recent application of this method to the Sachdev-Ye-Kitaev (SYK) model \cite{Sachdev_1993, SYK, Maldacena_rmkSYK}, similarly based on Koteck\'y and Preiss criteria.

    \item The spectral gap for $\CL$ implies, following the quasi-adiabatic framework \cite{chen2023efficient, BC26}, that the purified Gibbs state can be prepared by short-time evolution of a time-dependent quasi-local all-to-all Hamiltonian. It is conceivable that real-time L.R. bounds for such systems also give quasi-polynomial time simulation algorithms e.g. \cite[Section 4]{chen2023speed}, \cite{McDonough_2025}.
\end{enumerate}


\paragraph{Quantum mixing times.} There is a longstanding literature on the design and analysis of quantum Markov chains. \cite{kastoryano2016commuting, Bardet_2021, Bardet2024, capel2021modified,kochanowski2024rapid, capel2024quasi} studied commuting lattice systems, and proved such systems are fast (and rapid) mixing by relying on certain decay-of-correlation assumptions. \cite{temme2015fast, temme2016thermalizationtimeboundspauli, Ding:2024ltx} studied bounded degree stabilizer models, and proved rapid/fast-mixing by relying on the absence of ``free energy barriers'' in such systems. A variety of special cases can also be analyzed under certain quantum-to-classical mappings \cite{Hwang24, BCL24, BGL25, paezvelasco2025efficientsimplegibbsstate,basso_quantum_2025}.

Proofs of mixing times in non-commuting systems began with the work of \cite{rouze2024efficient}, who proved a spectral gap for the \cite{chen2023efficient} Lindbladian at high temperatures by relying on techniques from the study of gapped ground states (namely, the stability of the gap \cite{Michalakis2011StabilityOF}). This was later strengthened by \cite{rouze2024optimal, bakshi2025dobrushinconditionquantummarkov} to rapid-mixing on bounded degree Hamiltonians by means of quantum analogues of the Dobrushin condition, see also \cite{tong2025fermions, vsmid2025polynomial, vsmid2025rapid, bakshi2026externalfields} for applications to related systems. This framework, based on the decay of the oscillator norm (and its dual), was largely introduced in \cite{MajewskiZegarlinski1995QSDI, MajewskiZegarlinski1996QSDII, MajewskiOlkiewiczZegarlinski1998} as a quantum analog of the results of  Aizenman and Holley \cite{Aizenman1987}. We note that classically, these Dobrushin-style conditions  \cite{dobrushin1985completely, dobrushin1987completely} suffice to establish rapid-mixing for all-to-all classical systems at high-temperatures, and it is this gap that we attempt to close in this work. However, as mentioned, we fall short of proving \textit{rapid} mixing. 

Recently \cite{BC26} studied one-dimensional quantum systems, and proved such systems are fast-mixing at all temperatures. At the heart of their analysis was the introduction of a ``pseudo-Lindbladian'' generator $\CK$, which admits locality properties absent in the generator $\CL$ \cite{chen2023efficient}, at the cost of not generating a CPTP map (indeed, this will later be the reason why our results are limited to fast-mixing). In this paper, we revisit this pseudo-Lindbladian in the context of all-to-all systems, see \cref{sec:overview}.

\paragraph{Outlook on disordered systems.} As mentioned, the most interesting future direction lies in the study of mixing times for random quantum Hamiltonians, in particular quantum spin glasses \cite{AJBray_1980} and the Sachdev-Ye-Kitaev Hamiltonian \cite{Sachdev_1993, SYK, Maldacena_rmkSYK}. The results in this paper, much like the Dobrushin condition \cite{dobrushin1985completely, dobrushin1987completely} applied to classical spin glasses, are far too weak to address such systems. Indeed, in the classical Markov chain literature, proofs of mixing times for spin glasses came about over $20$ years after their counterparts for lattice systems; developed through a prominent framework known as \textit{spectral independence} \cite{ALG24,anari2021entropic, AJKPV22,Eldan2021} and variants. 

We hope that the techniques developed here for deterministic systems can nevertheless serve as a stepping stone, shedding some light on the relevant challenges in all-to-all quantum systems, as well as their connections to operator growth theory and classical simulation algorithms.

\subsection{Technical overview}\label{sec:overview}

\textbf{Organization.} We dedicate this section to an overview of our analysis. We refer the reader to \cref{section:prelim} for the relevant background on weighted inner products, detailed-balance, and Lindbladian dynamics. A cartoon of our argument is presented in \cref{fig:argument-summary}, which we elaborate on in the sequence.

\begin{enumerate}
    \item We begin in \cref{section:challenges} by posing a central challenge: on establishing quasi-locality properties for Lindbladian dynamics in all-to-all quantum systems, integral to the previous proofs of fast-mixing in the literature \cite{rouze2024efficient, rouze2024optimal, bakshi2025dobrushinconditionquantummarkov}.
    \item  In \cref{section:overview-cte}, we present a simple cluster expansion for the short (complex) time evolution of local operators under all-to-all Hamiltonians, based on \cite{Neto2004, Mann_2021}. 
    \item In \cref{section:overview-K}, we discuss how this convergence of complex time evolution naturally feeds into the pseudo-Lindbladian $\CK$ framework of \cite{BC26}, and we introduce a simplified Dobrushin-style condition for $\CK$; sufficient to prove a spectral gap.
    
    \item In \cref{section:overview-L-to-K}, we discuss how to lift the spectral gap for $\CK$ into a spectral gap for the original Lindbladian dynamics $\CL$ \cite{chen2023efficient}, by comparing their Dirichlet forms. 
\end{enumerate}

The technical proofs of points 2, 3, 4 are presented in \cref{section:quasi-locality,section:dobrushin,section:d-form-comparison} respectively.

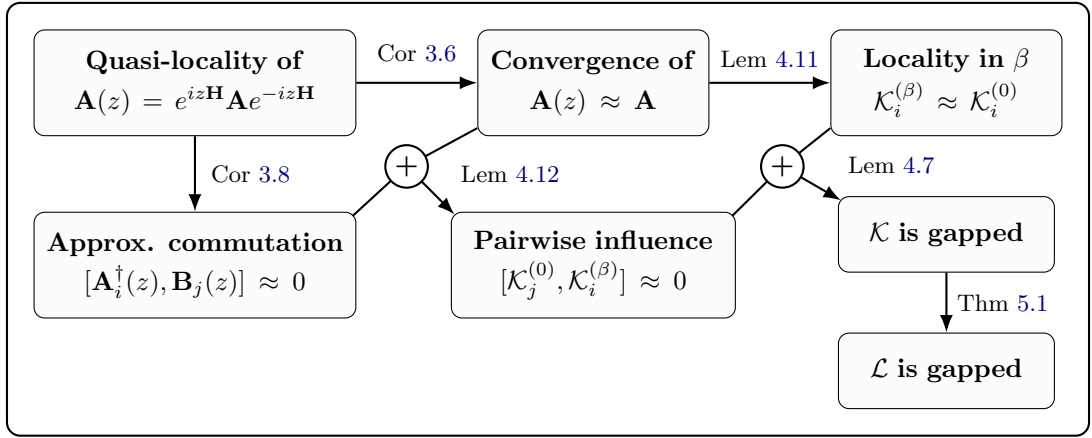
\begin{figure}[t]
\centering
\begin{tikzpicture}[
  box/.style={rounded corners, minimum width=.9cm, minimum height=1.4cm, text width=3.5cm, align=center, font=\small, draw=black, fill=gray!2},
  arrow/.style={-{Latex}, thick},
  node distance=1cm and 1.6cm,
  box2/.style={rounded corners, minimum width=1.0cm, minimum height=1.4cm, text width=4.0cm, align=center, font=\small, draw=black, fill=gray!2},
  box3/.style={rounded corners, minimum width=.8cm, minimum height=1.4cm, text width=2.8cm, align=center, font=\small, draw=black, fill=gray!2},
  box4/.style={rounded corners, minimum width=1.0cm, minimum height=1.0cm, text width=2.6cm, align=center, font=\small, draw=black, fill=gray!2},
    boundingbox/.style={draw=black, thick, rounded corners=5pt, inner sep=10pt, fit=#1}
]

\node[box2] (T1) {
  \textbf{Quasi-locality of} \\ \vspace{3pt}$\vA(z)= e^{iz\vH}\vA e^{-iz\vH}$
};

\node[box3, right =of T1] (T2) {
  \textbf{Convergence of} \\ \vspace{3pt}$\vA(z)\approx \vA$
};

\node[box2, below=of T1] (T3) {
  \textbf{Approx. commutation} \\ \vspace{3pt}$[\vA_i^\dagger(z), \vB_j(z)]\approx 0$
};

\node[box3, right=of T2] (T4) {
  \textbf{Locality in $\beta$} \\\vspace{3pt}
  $\CK_i^{(\beta)}\approx \CK_i^{(0)}$
};

\node[box, below=of T2] (T5) {
 \textbf{Pairwise influence} \\ \vspace{3pt} $[\CK_j^{(0)},  \CK_i^{(\beta)}]\approx 0$
};

\node[box4](T6) at ($(T4.south) + (0,-1.3cm)$) {
  \textbf{$\CK$ is gapped}\\[3pt]
};

\node[box4] (T7) at ($(T6.south) + (0,-1.3cm)$) {
  \textbf{$\CL$ is gapped}
};

\draw[arrow] (T1.east) -- (T2.west) node[midway, above = 4pt] {\footnotesize Cor~\ref{cor:cte-conv}};

\draw[arrow] (T1) -- (T3) node[midway, right = 2pt] {\footnotesize Cor~\ref{cor:total-corr}};
\draw[arrow] (T2) -- (T4) node[midway, above = 2pt] {\footnotesize  Lem~\ref{lemma:Kt-approx}};

\draw[arrow] (T6.south) -- (T7.north) node[midway, right] {\footnotesize Thm~\ref{theorem:K-to-ckg}};

\coordinate (mergepoint1) at ($ (T1)!0.5!(T3) + (2.8cm,0.1cm) $);
\coordinate (corner2) at ($ (T3.east) + (-0.05cm,0.65cm)$);
\draw[-, thick] (corner2) -- (mergepoint1);
\coordinate (corner3) at ($ (T2.west) + (0,-0.6cm)$);
\draw[-, thick] (corner3) -- (mergepoint1);
\draw[arrow] (mergepoint1) -- ($ (T5.west) + (0.05cm,.67cm)$) node[midway, above=-1pt, xshift=30pt] {\footnotesize Lem~\ref{lemma:Kmf-approx}}; 
\node[draw,circle,inner sep=1.5pt,thick,fill=white] at (mergepoint1) {\small\textbf{+}};


\coordinate (mergepoint2) at ($ (T3)!0.5!(T4) + (2.8cm,.1cm) $);
\coordinate (corner4) at ($ (T5.east) + (-0.05,0.65cm)$);
\draw[-, thick] (corner4) -- (mergepoint2);
\coordinate (corner5) at ($ (T4.west) + (0,-0.6cm)$);
\draw[-, thick] (corner5) -- (mergepoint2);
\draw[arrow] (mergepoint2) -- ($ (T6.west) + (0.04cm,.45cm)$) node[midway,  above=0pt, xshift=30pt] {\footnotesize Lem~\ref{lemma:d-conditions}};
\node[draw,circle,inner sep=1.5pt,thick,fill=white] at (mergepoint2) {\small\textbf{+}};

\node[boundingbox=(T1)(T2)(T3)(T4)(T5)(T6)(T7)] {};

\end{tikzpicture}

\caption{
An outline of the proof of \cref{theorem:all-to-all-mixing}.}
 \label{fig:argument-summary}

\end{figure}

\newpage

\subsubsection{A central challenge}
\label{section:challenges}

Central to existing proofs of fast or rapid mixing for the \cite{chen2023efficient} Lindbladian dynamics \cite{rouze2024efficient, rouze2024optimal, vsmid2025rapid, vsmid2025polynomial, tong2025fermions, BC26} is to establish basic quasi-locality properties for the operator Fourier transform of a local operator $\vA$:
\begin{align}
    \hat{\vA}(\omega):=\frac{1}{\sqrt{2\pi}}\int_{-\infty}^\infty f_\sigma(t)\cdot e^{-i\omega t}\cdot e^{i\vH t}\vA e^{-i\vH t} \rd t
\end{align}
where $f_\sigma(t)$ is a Gaussian filter function, see \cref{section:prelim} for further background. These operators are the ``updates'' or transitions of the Lindbladian $\CL$ \eqref{eq:lind-eq-intro}, and in prior work for lattice systems, Lieb-Robinson bounds were used to argue that they can be written as a linear combination of Pauli operators whose norm decays with the distance of said operators to the support of $\vA$.

\begin{question}\label{question:ahat}
    Can one establish such a ``shell decomposition'' for $\hat{\vA}(\omega)$ in all-to-all systems?
\end{question}

We remark that the quasi-locality of $\hat{\vA}(\omega)$ in lattice systems holds at \textit{all temperatures}; here, even allowing the Gaussian filter width $\sigma$ to depend mildly on $\beta$ (so, a high-temperature statement), it remains unclear how to address \cref{question:ahat}. In the absence of basic locality properties for $\CL$, we are forced to perform a roundabout approach to establishing a spectral gap, where we first replace $\CL$ by a strictly more localized Markovian dynamics. The starting point to this approach is to study the \textit{short} complex-time evolution, which we discuss in the sequence. 

\subsubsection{Quasi-locality of complex-time evolution.}
\label{section:overview-cte}

In \cref{section:quasi-locality}, we derive a series of quasi-locality properties tied to the convergence of complex time evolution $z\in \mathbb{C}:\vA(z):=e^{iz\vH}\vA e^{-iz\vH}$ of local operators $\vA$ at high-temperatures (i.e. small $|z|$), which roughly express $\vA(z)$ as a linear combination (or superposition) of \textit{sparse} Pauli operators. The following cluster expansion is based on minor modifications to the expansion in \cite{Neto2004, Mann_2021} done for the partition function.

\begin{lemma}
    [Cluster expansion for $\vA(z)$, informal] The complex-time evolution of a local operator $\vA$ on site $j\in [n]$ under an all-to-all Hamiltonian $\vH$ (\cref{def:all-to-all}) satisfies the expansion:
    \begin{equation}
    e^{iz\vH}\vA e^{-iz\vH}-\vA := \sum_{\substack{\sF\subseteq \Gamma \\ \sF \textsf{ connected to }j}} \vA_{\sF}(z), \quad \|\vA_{\sF}(z)\|\leq \prod_{e\in \sF}\big(e^{2|z|\|\vh_e\|}-1\big),
\label{eq:intro-cte}
\end{equation}
\end{lemma}
\noindent which, roughly speaking, expresses the contribution to $\vA(z)$ associated to a connected cluster of interactions $\sF\subseteq \Gamma$ as a function of the product of interaction strengths in $\sF$. 

In some intuitive sense, \eqref{eq:intro-cte} suggests that information from $j$ \textit{only spreads to few other sites at a time}. To make this quantitative, we introduce two quasi-locality parameters which later play a central role in our mixing-time analysis. 

\begin{definition}\label{def:cte-intro}
    We refer to the \textit{complex-time evolution constant} $\mathsf{conv}_\beta$ as the maximum deviation of any single-site Pauli operator from itself  under complex-time evolution:
    \begin{equation}
        \mathsf{conv}_\beta:=\max_{a\in \CS_{[n]}^1}  \sup_{|z|\leq \beta}\|\vA^a(z)-\vA^a\| \label{eq:cte-intro}
    \end{equation}
\end{definition}

We remark that the reason this constant will play such a central role in the latter stages of this paper, is due to liberal use of KMS H\"older's inequality, see \cref{lem:kms-holder}.

\begin{definition}\label{def:pairwise-corr-intro}
    We refer to the \textit{pairwise correlations constant} $\mathsf{corr}_\beta$ as the maximum commutator-norm of any pair of single-site Pauli operators (on distinct sites) under complex-time evolution:
    \begin{equation}
        \mathsf{corr}_\beta(i, j):=\max_{\substack{a\in \CS_{i} \\ b\in \CS_j}}  \sup_{|z|\leq \beta}\|[\vA^a(z)^\dagger, \vB^b(z)]\|, \quad \text{and}\quad \mathsf{corr}_\beta:=\max_{i\in [n]} \sum_{j\in [n]} \mathsf{corr}_\beta(i, j)
    \end{equation}
\end{definition}

Intuitively, \cref{def:pairwise-corr-intro} captures how much of the support of $\vA(z)$ has percolated to site $j$; and $\mathsf{corr}_\beta$ is a measure of the total influence of $i$ on the other sites. Next, we discuss how fast-mixing can be reduced directly to these parameters.

\begin{remark}
    In \cref{section:cluster-expansion}, we argue the definitions above can be reduced to series expansions akin to \eqref{eq:intro-cte}, and we express sufficient conditions for the convergence of the series at small $\beta$. 
\end{remark}

\subsubsection{Dobrushin conditions for the pseudo-Lindbladian $\CK$ \cite{BC26}.}
\label{section:overview-K}

The convergence of complex-time evolution enables us to turn to the ``pseudo-Lindbladian'' of \cite{BC26}, who designed a superoperator $\CK = \sum_{i\in [n]} \CK_i$  resembling the \cite{chen2023efficient} Lindbladian, but with significantly sharper quasi-locality properties instrumental in proving spectral gaps.

\begin{definition}
    For every single-site Pauli operator $\vA^a\in \CS_{[n]}^{1}$, we introduce the superoperator:
    \begin{equation}\label{eq:CK}
 \CK_a[\vX]:=
  [\vA^{a},\vX]\,
  \bigl(\vrho^{1/2}\vA^{a\dagger}\vrho^{-1/2}\bigr)
  \;-\;
  \bigl(\vrho^{-1/2}\vA^{a\dagger}\vrho^{1/2}\bigr)
  [\vA^{a},\vX],
\end{equation}
and further let $\CK_j:= \sum_{a\in \CS_j^1}\CK_a$ be the generator defined by all single-site Pauli jumps on $j\in [n]$.
\end{definition}

 This generator satisfies a series of nice properties:

\begin{lemma}
    [Properties of $\CK$ {\cite[Lemma III.5]{BC26}}]\label{lemma:K-properties}
    Each $\CK_j$ is $\vrho$-detailed-balanced and is negative semi-definite. Furthermore, the kernel of $\CK_j$ is exactly the set of locally trivial operators:
    \begin{equation}
        \CK_j[\vX]=0\iff \vX=\vX_{[n]\setminus j}\otimes \mathbb{I}_j\label{eq:kernel-K}
    \end{equation}
\end{lemma}
We emphasize that this explicit, local characterization of the kernel of the single-site generators $\CK_i$ \eqref{eq:kernel-K}, is not true for the \cite{chen2023efficient} generator $\CL_i$. We refer to this superoperator as a \textit{pseudo-Lindbladian} as it admits non-positive spectra but does not generate a CP map, and thus doesn't define a physical channel one could implement on a quantum computer. Nevertheless, shortly we show why a spectral gap for $\CK$ implies a spectral gap for $\CL$ (\cref{section:overview-L-to-K}). \\

\noindent \textbf{Dobrushin conditions for $\CK$.} We now establish a set of sufficient conditions for $\CK$ to admit a spectral gap. Akin to the previous Dobrushin-style conditions in the quantum Gibbs sampling literature \cite{rouze2024optimal, bakshi2025dobrushinconditionquantummarkov, MajewskiOlkiewiczZegarlinski1998}, these conditions intuitively capture two locality properties of the local generators $\CK$: locality in \textit{temperature}, and locality in \textit{space}.\\

\noindent \textit{Locality in temperature.} We assume that for each site $j\in [n]$ the single-site generator $\CK_j^{(\beta)}:=\CK_j$ can be approximated by that at infinite-temperature (the depolarizing semigroup):
\begin{equation}
    \forall j\in [n]:\quad  \|\CK_j^{(\beta)}-\CK_j^{(0)}\|_{\vrho}\leq \eta_j.\label{eq:overview-temp}
\end{equation}
\noindent \textit{Pairwise influence.} We assume that for each site $j\in [n]$, the generator $\CK_j$ approximately commutes with the depolarizing semigroup on another site $\ell\neq j$:
    \begin{equation}
       \forall j\neq \ell:\quad  \|[\CK^{(0)}_\ell, \CK_j^{(\beta)}]\|_{\vrho}\leq \kappa_{j\ell}.\label{eq:overview-space}
    \end{equation}

    \begin{remark}
    The intuition behind \eqref{eq:overview-space} is that at high-temperatures, a single update $\CK_j$ on $j\in [n]$ cannot simultaneously influence all the other $\ell\neq j;$ resembling a mean-field statement. 
\end{remark}

    Shortly, we provide further intuition on why both of these conditions hold for a suitable choice of constants $\{\eta_j\}_j, \{\kappa_{j\ell}\}_{j, \ell}$ (and scale roughly linearly in $\beta$) at sufficiently high temperatures. In \cref{section:dob-conditions}, we prove that they define a simple sufficient condition for fast-mixing:

\begin{lemma}
    [A Dobrushin condition for $\CK$, informal]\label{lemma:intro-dob-cond} There exists a universal constant $\lambda$ satisfying the following guarantee. Assume \eqref{eq:overview-temp},\eqref{eq:overview-space} hold for a choice of constants  $\{\eta_j\}_j, \{\kappa_{j\ell}\}_{j, \ell}$, as well as the convergence of complex-time evolution \eqref{eq:cte-intro} with constant $\mathsf{conv}_\beta$. 
    
    Then, if said constants satisfy the inequality:
    \begin{equation}
        \gamma:= \max_j \big[2(1+\mathsf{conv}_\beta)^2 \cdot \eta_j + \sum_{\ell\neq j} \kappa_{j\ell}\big] < \lambda,
    \end{equation}
    then $\CK$ admits a system-size independent spectral gap: $\mathsf{gap}(\CK)\geq \lambda-\gamma$.
\end{lemma}

The proof of \cref{lemma:intro-dob-cond} is largely based on revisiting the argument in \cite{MajewskiOlkiewiczZegarlinski1998}, however, with two central modifications. 
\begin{enumerate}
    \item the kernel condition \eqref{eq:kernel-K} enables us to state this quantum Dobrushin condition as a pairwise influence criteria, as opposed to a tri-partite criteria as in \cite{rouze2024optimal, bakshi2025dobrushinconditionquantummarkov}.
    \item Since $\CK$ does not generate a CP channel, $e^{t\CK}$ is not necessarily contractive in the induced $\|\cdot \|_{\infty\rightarrow \infty}$ norm. Accordingly, we re-route the argument under a modification of the oscillator norm we coin the \textit{KMS oscillator norm}; which could be of independent interest. 
\end{enumerate}

\noindent \textbf{Satisfying the conditions via operator growth bounds.} It only remains to argue why these influence assumptions \eqref{eq:overview-temp},\eqref{eq:overview-space} hold for $\CK$ at high-temperatures. The idea is to use the explicit structure of $\CK$ in \eqref{eq:CK} to attempt to reduce these criteria to the previously computed operator growth bounds. The following statement quantifies this reduction in the pairwise influence case:

\begin{lemma}
    [Pairwise influences, informal]
    For any pair of sites $\ell\neq j\in [n]$, the generator $\CK_j$ at inverse-temperature $\beta$ satisfies the approximation
    \begin{equation}
          \|[\CK_\ell^{(0)}, \CK_j^{(\beta)}]]\|_{\vrho} \leq c_\sq\cdot (1+\mathsf{conv}_{\beta/4})^2\cdot \mathsf{corr}_{\beta/4} (j, \ell)
    \end{equation}
    with $c_\sq$ some explicit constant that depends on $\sq$. 
\end{lemma}

The case of the locality in temperature condition is analogous. By combining the cluster expansions, with the approximations for $\CK$, we conclude that at sufficiently high temperatures the conditions in \cref{lemma:intro-dob-cond} are satisfied and therefore $\CK$ is gapped (recall \cref{fig:argument-summary}). 

\subsubsection{Dirichlet form comparison arguments }
\label{section:overview-L-to-K}

Now equipped with the Poincar\'e inequality for $\CK$, it remains to argue why this implies a Poincar\'e inequality for $\CL$. In the recent work of \cite{BC26} for one-dimensional systems, such a comparison was done by relating their Dirichlet forms:
\begin{equation}
  \forall \vX:\quad   \langle \vX, -\CL^\dagger[\vX]\rangle_{\vrho} \geq \eta_\beta\cdot  \langle \vX, -\CK[\vX]\rangle_{\vrho}\label{eq:overview-D-form-comp}
\end{equation}
up to a temperature dependent constant $\eta_\beta$. Although this is broadly the strategy we employ in this paper as well, their argument made liberal use of the unconditional convergence of complex time-evolution and sharp bounds on the convergence of real-time evolution in 1D; which we currently do not know how to implement in all-to-all systems (\cref{question:ahat}). 

\begin{remark}
    In particular, their comparison argument relied on a shell decomposition for the operator Fourier transform $\hat{\vA}(\omega)$ which decays both in space and in $\omega$, see \cite[Lemma VI.2]{BC26}.
\end{remark}

Arguably the most technical part of our work is a proof of a Dirichlet form comparison akin to \eqref{eq:overview-D-form-comp} at sufficiently high temperatures, see \cref{section:d-form-comparison}. Very roughly speaking, the proof strategy is to understand $\CL$ as a Markov process with updates given by the operator Fourier transform $\hat{\vA}(\omega, t)$, and to sequentially compare Dirichlet forms with the analogous Markov processes with updates $\vA(t)$ and $\vA$ respectively (i.e. $\omega, t$ have been ``peeled off''). \\

\noindent \textit{The key step behind \eqref{eq:overview-D-form-comp}.} To give concrete intuition, we highlight a key step behind how we relate the Dirichlet forms of generators with jump operators $\vA$ (single-site operators) and $\vA(t)$ (time-evolved single-site operators). We begin with Duhamel's formula:
\begin{equation}
        \vA(t) - \vA = i\int_0^t e^{i\vH s} [\vH, \vA]e^{-i\vH s} \rd s.
    \end{equation}
By the reverse triangle inequality, for any operator $\vX$, we can then (very roughly speaking) understand the norms of commutators of $ \vA(t)$ with $\vX$, as a function of $\vA$ and $\vX$, up to an error term which depends on commutators of 2-site operators:
\begin{align}
    \underbrace{\|[\vA(t), \vX]\|_{\vrho}}_{\textsf{``Dirichlet form''  of }\CL_a} \geq \underbrace{\|[\vA, \vX]\|_{\vrho}}_{\textsf{``Dirichlet form''  of }\CK_a} - \int_0^t \underbrace{\bigg\| \big[ [\vH, \vA](s), \vX\big]\bigg\|_{\vrho}}_{\substack{\textsf{``Dirichlet form''  of }\CL\\ \textsf{under 2-site updates}}} \rd s
\end{align}
Indeed, $[\vH, \vA]$ is a linear combination of $2$-local Paulis, and thus the second term on the RHS above can be interpreted as the Dirichlet form of a generator with 2-site updates (as opposed to single-site). 

\begin{remark}
    We use the quotation marks as the Dirichlet form of such generators are technically integrals over the \textit{squared} commutator-norm; here we simply mean to convey the intuition. 
\end{remark}

The key idea, stemming from recent work \cite{chen2025GibbsMarkov, chen2025learning, BC26, BCV25}, is to perform a \textit{path comparison argument} akin to those in the classical literature \cite{Diaconis1993COMPARISONTF}: any Markov chain with pairwise updates should be simulable by a Markov chain with single-site updates (up to a $\beta$-dependent constant loss). Here we implement an analogous strategy by relying on the convergence of complex-time evolution (\cref{def:cte-intro}) at high-temperatures, to relate their Dirichlet forms. 

To fully prove \eqref{eq:overview-D-form-comp} we need to similarly argue how to ``peel off'' $\omega$ in order to relate $\hat{\vA}(\omega, t)$ and $\vA(t)$. This step is technical but ultimately based on the same intuition above. We refer the reader to \cref{section:d-form-comparison} for further discussion.

\subsection*{Acknowledgements}



 After writing a preliminary version of the manuscript, ChatGPT 5.5 (Extended Pro) informed the author their initial cluster expansion calculations had already been performed and generalized by Neto\v{c}n\'y and Redig \cite{Neto2004}. The author further thanks the model for simplifying the (elementary) combinatorial proofs in \cref{section:cte-conv}, on sufficient conditions for the convergence of said cluster expansion.

\section{Preliminaries}
\label{section:prelim}

\noindent \textbf{Notation.} We consider quantum spin systems of finite size $n$, each of constant local dimension $2^\sq$. Each such spin $i\in [n]:=\{1, 2, \cdots, n\}$ will be referred to as a \textit{site}. The associated Hilbert space is denoted as $\calH = \otimes_{i=1}^{n} (\BC^{2^\sq})$. The space of linear operators on $\calH$ is referred to as $\CB(\calH)$; we use boldface letters to refer to operators. A superoperator $\CN:\CB(\calH)\rightarrow \CB(\calH)$ is a map between linear operators; we use in curly font~$\CL$ with matrix arguments in square brackets to refer to $\CL[\vrho]$. We refer to $[\cdot, \cdot]$ and $\{\cdot, \cdot\}$ as the commutator and anti-commutator respectively.

We denote the set of single-site $\sq$-weight Pauli strings, supported on the $i$th $2^\sq$ dimensional qudit as:
\begin{align}
    \CS_i &:= ( \{\vI_2,\vX,\vY,\vZ\}^{\otimes \sq}\} )\otimes \vI_{2^\sq}^{n-1},
\end{align}
and we will never care how the Pauli strings are actually embedded within a site. 

Sans serif is used to refer to subsets of sites $\sA\subseteq [n]$. For any subset $\sA\subseteq [n]$, we define the set of single-site Pauli strings, and multi-site Pauli strings:
\begin{align}
    \CS^1_\sA &:= \cup_{i\in \sA} \CS_i,\qquad \CS_\sA := ( \{\vI_2,\vX,\vY,\vZ\}^{\otimes \sq\labs{\sA}})\otimes \vI_{2^\sq}^{n-\labs{\sA}}.
\end{align}

We use $\CO(\cdot),\Omega (\cdot)$ to denote asymptotic upper and lower bounds; and in an abuse of notation $ a \lesssim b$ to mean $a \le c b$ for a universal constant $c$.

\subsection{The KMS inner product and detailed-balance}

We dedicate this section to the relevant non-commutative weighted inner product. We assume $\vrho$ is full rank throughout this section.

\begin{definition}
    [KMS inner product]\label{defn:s_inner} Given a full-rank state $\vrho$, we define the KMS or $\vrho$-inner product of two operators $\vX, \vY$ as:
    \begin{equation}
        \langle \vX,\vY\rangle_{\vrho}:=\tr[\vX^\dagger \vrho^{\frac{1}{2}}\vY\vrho^{\frac{1}{2}}]\,.
    \end{equation}
    We denote by $\|\vX\|_{\vrho} = \sqrt{\langle \vX,\vX\rangle_{\vrho}}$ the $\vrho$-weighted $2$-norm.
\end{definition}

Associated to a weighted inner product, we can define a notion of detailed-balance:
\begin{definition}
    [Detailed-balance]\label{eq:detailed-balance} A superoperator $\CN$ is said to be $\vrho\mathsf{-detail-balanced}$ if it is self-adjoint w.r.t. the $\vrho$ inner product:
    \begin{equation}
        \forall \vX, \vY:\quad \langle \vX, \CN[\vY]\rangle_{\vrho}= \langle \CN[\vX], \vY\rangle_{\vrho}
    \end{equation}
\end{definition}

We make extensive use of the following version of H\"older's inequality for the $\vrho$ norm.

\begin{lemma}[H\"older in KMS Norm, e.g., {~\cite[Lemma IX.4]{chen2025GibbsMarkov}}]\label{lem:kms-holder}
For any pair of operators $\vX, \vY$, and full rank state $\vrho,$
\begin{align}
\|\vX\vY\|_{\vrho} &\le \|\vrho^{1/4}\vX\vrho^{-1/4}\| \cdot \|\vY\|_{\vrho}
\end{align}     
\end{lemma}

Of course, the efficacy of the above is contingent on the convergence of complex-time evolution of $\vX$. We also consider the following induced norm on superoperators:

\begin{definition}
    [KMS-Induced Superoperator Norm]\label{defn:induced_norm} For any superoperator $\CN:\CB(\calH)\rightarrow \CB(\calH)$ and full-rank state $\vrho$,
    \begin{align}
    \norm{\CN}_{\vrho}:= \sup_{\vO}\frac{\norm{\CN[\vO]}_{\vrho}}{\norm{\vO}_{\vrho}}.
\end{align}

\end{definition}

\subsection{The \texorpdfstring{\cite{chen2023efficient}}{[CKG23]} Lindbladian}

We dedicate this section to an introduction on the construction of \cite{chen2023efficient}. We begin with the truly invaluable operator Fourier transform:

\begin{definition}
    The operator FT of an operator $\vA$ at energy width $\sigma>0$, associated to the Hamiltonian $\vH$, is written as: 
    \begin{align}
\hat{\vA}_{\sigma}(\omega)=  \frac{1}{\sqrt{2\pi}}\int_{-\infty}^{\infty} e^{i \vH t} \vA e^{-i \vH t} e^{-i \omega t} f_{\sigma}(t)\rd t.\label{eq:OFT}
    \end{align}
\end{definition}
\noindent In the above, the function $f_\sigma(t)$ above is a Gaussian filter:
    \begin{align}
       f_{\sigma}(t) = e^{-\sigma^2t^2}\sqrt{\sigma\sqrt{2/\pi}}\quad \text{and}\quad  \hat{f}_{\sigma}(\omega)=\frac{e^{- \omega^2/4\sigma^2}}{\sqrt{\sigma\sqrt{2\pi}}}  = \frac{1}{\sqrt{2\pi}}\int_{-\infty}^{\infty}e^{-i\omega t} f_\sigma(t)\mathrm{d}t.\label{eq:fwft}
    \end{align}

\noindent Whenever implicit, we omit the subscripts ${\hat{\vA}_{\sigma}}(\omega)\equiv {\hat{\vA}}(\omega)$, $f(t)=f_\sigma(t)$. We are now in a position to define the family of Lindbladians of \cite{chen2023efficient}.

\begin{definition}
    [The \cite{chen2023efficient} Lindbladian]\label{def:ckg-lind} Fixed $\vH$, $\beta>0$, $0<\sigma\leq \beta^{-1}$, and a single self-adjoint jump $\vA^a = \vA^{a\dagger}$. The Lindbladian~\cite{chen2023efficient} is written as
\begin{align}
		\CL_a[\cdot] = \underset{\text{``coherent''}}{\underbrace{ -i [\vC^a, \cdot]}} + 
		\int_{-\infty}^{\infty} \gamma(\omega) \bigg(\underset{\text{``transition''}}{\underbrace{\hat{\vA}_\sigma^a(\omega)(\cdot)\hat{\vA}_\sigma^{a}(\omega)^\dagger}} - \underset{\text{``decay''}}{\underbrace{\frac{1}{2}\{\hat{\vA}_\sigma^{a}(\omega)^\dagger\hat{\vA}_\sigma^a(\omega),\cdot\}}}\bigg)\rd\omega,\label{eq:exact_DB_L}
\end{align}
\noindent the ``coherent part'' $\vC^{a}$ is a Hermitian operator:
\begin{align}
    &\vC^a = \iint_{-\infty}^{\infty} \gamma(\omega) \cdot c(t) \cdot \hat{\vA}_\sigma^a(\omega,t)^{\dagger}\hat{\vA}_\sigma^a(\omega,t)  \rd t\rd \omega, \quad  \\\text{with}\quad  &c(t) := \frac{1}{\beta\sinh(2\pi t/\beta)} \quad \text{and}\quad  \hat{\vA}_\sigma^a(\omega,t):=e^{i\vH t}\hat{\vA}_\sigma^a(\omega)e^{-i\vH t}.
\end{align}
\end{definition}
\noindent The choice of $\gamma$ ensures detailed-balance; here we  use the shifted-Metropolis weight function:
\begin{align}
    \gamma_\mathsf{M}(\omega) := \exp\bigg(-\beta\max\left(\omega +\frac{\beta \sigma^2}{2},0\right)\bigg).\label{eq:Metropolis}
\end{align} 

\begin{lemma}
    [{\cite[Corollary A.1]{chen2025GibbsMarkov}}]
    When $\norm{\vA^a}=1$, we have system-size independent norm bounds $\norm{\CL_a}_{1-1}\le O(1+\beta)$ under $\gamma_\mathsf{M}$.
\end{lemma}

\subsection{The spectral gap}

Our goal in this work is to prove a spectral gap for the \cite{chen2023efficient} Lindbladian family, which quantifies the rate of decay of the variance of observables:

\begin{definition}[Spectral gap]
\label{def:spectral_gap}
The spectral gap of a primitive KMS-$\vrho$ detailed-balanced Lindbladian
$\CL$ is
\begin{align}
    \lambda_{\mathrm{gap}}(\CL)
    : =
    \inf_{\vX}
    \frac{\langle\vX,-\CL^\dagger[\vX]\rangle_{\vrho}}
         {\norm{\vX-\vI\tr[\vrho\vX]}_{\vrho}^2}.
    \label{eq:gap_def}
\end{align}
\end{definition}
The numerator above is known as the Dirichlet form; see \cref{section:d-form-comparison} for a discussion. The denominator is simply the $\vrho$ weighted variance of $\vX$. We note that the spectral gap already provides a means to control the rate of convergence to equilibrium, albeit at the cost of a dimension-dependent factor:

\begin{lemma}[Mixing from the spectral gap {\cite{Kastoryano2012QuantumLS}}]
\label{lem:mixing_from_gap}
Let $\CL$ be primitive and KMS-$\vrho$ detailed-balanced. Then the mixing time of $\CL$ satisfies
\begin{align}
    t_{\mathrm{mix}}(\CL)
   \leq
    \frac{\log\!\left(2\|\vrho^{-1/2}\|\right)}
         {\lambda_{\mathrm{gap}}(\CL)}.
    \label{eq:mixing_time_def}
\end{align}

\end{lemma}

\section{Quasi-locality of complex-time evolution}
\label{section:quasi-locality}

In this section we collect several locality properties, broadly regarding the convergence of complex time evolution, for a generic dense $\sk$-local Hamiltonian (\cref{def:all-to-all}):
\begin{equation}
    \vH=\sum_{e\in \Gamma} \vh_e, \qquad \mathsf{supp}(\vh_e)\subseteq e
\end{equation}
The set $\Gamma\subseteq {[n] \choose \sk}$ is the set of (hyper-)edges. We organize this section as follows. 

\begin{enumerate}
    \item In \cref{section:cluster-expansion}, we give a simple cluster expansion (\cref{lemma:cluster}) for the complex-time evolution $\vA(z):=e^{iz\vH}\vA e^{-iz\vH}$ of local Pauli operators; based on \cite{Neto2004} for the partition function.
    \item In \cref{section:cte-conv}, we discuss conditions for the convergence of said cluster expansion and discuss corollaries. 
\end{enumerate}

\subsection{Cluster expansions for the complex-time evolution}
\label{section:cluster-expansion}

To state the cluster expansion, we require some basic graph-theoretic definitions. First, the local Hamiltonian $\vH$ naturally induces a graph on the Hamiltonian terms. 

\begin{definition}
    [Interaction Graph]
    We refer the \textit{interaction graph} $G_\sI$ of $\vH$ the graph whose vertex-set is $\Gamma$, and two terms $e_1, e_2\in \Gamma$ are connected in $G_\sI$ if $\mathsf{supp}(e_1)\cap \mathsf{supp}(e_2)\neq \emptyset$.
\end{definition}

The cluster expansion will group terms by connectivity. We refer to a connected cluster in the interaction graph $G_\sI$ as a set of edges / Hamiltonian terms (without multiplicities) which is strongly connected in  $G_\sI$ (meaning, no isolated connected components).

\begin{definition}
    [Connected Clusters] A subset of the edges $\sF\subseteq \Gamma$ is said to be a \textit{connected cluster} $\sF\in \mathsf{CC}$ if $\sF$ is strongly-connected in $G_\sI$. 
\end{definition}

\begin{remark}
    In an abuse of notation, we say a site $i\in [n]$ is in $\sF\subseteq \Gamma$ if $\exists e\in \sF$ s.t. $i\in \mathsf{supp}(\vh_e)$.
\end{remark}

We are now in a position to state the cluster expansion for the complex-time evolution. The below formulation is largely inspired by shell or annulus decompositions from the literature of Lieb-Robinson bounds in lattice Hamiltonians.

\begin{lemma}
    [Cluster expansions for the complex-time evolution]\label{lemma:cluster}
    Fix $z\in \BC$, and let $\vA\in \CS_i$ be a single-site Pauli operator supported on a fixed site $i\in [n]$. The complex time evolution $\vA(z):=e^{iz\vH}\vA e^{-iz\vH}$ of $\vA$ admits the following cluster expansion:
    \begin{align}
        \vA(z) :=\vA+\sum_{\substack{\sF\in \mathsf{CC}\\ i\in \sF}}\vA_{\sF}(z), \quad \text{where}\quad \|\vA_{\sF}(z)\|\leq \prod_{e\in \sF} \big(e^{2|z|\|\vh_e\|}-1\big)
    \end{align}
    where $\mathsf{supp}(\vA_\sF) = \mathsf{supp}(\sF)\subseteq [n]$.
    
\end{lemma}

Namely, $\vA(z)$ is decomposed into a linear combination of operators supported only on the connected clusters; in such a manner that the strength / norm of the contribution of each cluster $\sF$ roughly decays with the product of interaction strengths within $\sF$. 

\subsubsection{Convergence of complex-time evolution and correlations}

Before we present the proof of \cref{lemma:cluster}, we present two simple corollaries. These corollaries express bounds on the convergence of complex-time evolution, as a function of properties of the underlying induced graph $G_\sI$ and its weights. We later leverage these definitions directly in our proof of the Dobrushin condition. These definitions are restatements of those presented in our technical overview in \cref{sec:overview}.

\begin{definition}\label{def:cte}
    We refer to the \textit{complex-time evolution constant} $\mathsf{conv}_\beta$ as the maximum deviation of any single-site Pauli operator from itself  under complex-time evolution:
    \begin{equation}
        \mathsf{conv}_\beta:=\max_{a\in \CS_{[n]}^1}  \sup_{|z|\leq \beta}\|\vA^a(z)-\vA\| 
    \end{equation}
\end{definition}

A bound on this constant is immediately offered by applying the triangle inequality to \cref{lemma:cluster}:

\begin{corollary}
    [Convergence of complex-time evolution]\label{cor:cte-conv} In the context of \cref{lemma:cluster}:
    \begin{align}
        \mathsf{conv}_\beta \leq \max_{i\in [n]} \sum_{\substack{\sF\in\mathsf{CC}\\ i\in\sF}}\prod_{e\in \sF} \big(e^{2|\beta|\|\vh_e\|}-1\big).
    \end{align}
\end{corollary}

We will shortly discuss conditions under which this sum converges. Next, we introduce a measure of pairwise correlations under complex-time evolution.

\begin{definition}\label{def:pairwise-corr}
    For any pair of sites $i\neq  j\in [n]$, we refer to the \textit{pairwise correlations constant} $\mathsf{corr}_\beta(i, j)$ as the maximum commutator-norm of any pair of single-site Pauli operators under complex-time evolution:
    \begin{equation}
        \mathsf{corr}_\beta(i, j):=\max_{\substack{a\in \CS_{i} \\ b\in \CS_j}}  \sup_{|z|\leq \beta}\|[\vA^a(z)^\dagger, \vB^b(z)]\|, \quad \text{and}\quad \mathsf{corr}_\beta:=\max_{i\in [n]} \sum_{j\neq i} \mathsf{corr}_\beta(i, j)
    \end{equation}
    Observe the above is $0$ when $z=0$ and depends only on $\mathsf{Im}(z)$.
\end{definition}

\begin{corollary}
    [Convergence of total pairwise correlations]\label{cor:total-corr} In the context of \cref{lemma:cluster}, 
    \begin{align}
      \sum_{j\neq i}  \mathsf{corr}_\beta(i, j) \leq 2 \sum_{\substack{\sF\in\mathsf{CC}\\ i\in\sF}} |\mathsf{supp}(\sF)|\cdot \prod_{e\in \sF} \big(e^{4|\beta|\|\vh_e\|}-1\big)
    \end{align}
    with $|\mathsf{supp}(\sF)|$ the number of distinct sites incident on $\sF$.
\end{corollary}

\begin{proof}
    We restrict our attention to $z = \mathsf{Im}(z)$ and a pair of Pauli operators $\vA, \vB$ on sites $i,j$ respectively. Let us shorthand $w_e:=e^{2|z|\|\vh_e\|}-1$ for simplicity. The cluster expansion in \cref{lemma:cluster} then guarantees:
    \begin{align}
        \|[\vA(z)^\dagger, \vB(z)]\|  \leq  & 2\sum_{\substack{\sF\in \mathsf{CC}\\ i, j\in \sF}} \bigg(\|\vA_{\sF}(z)^\dagger\|+\|\vB_{\sF}(z)^\dagger\|\bigg) \\ &+
        2\sum_{\substack{\sT\in \mathsf{CC}\\ j\in \sT}}\sum_{\substack{\sF\in \mathsf{CC}\\ i\in \sF}} \|\vA_{\sF}(z)^\dagger\|\cdot \|\vB_{\sT}(z)\| \cdot \mathbb{I}[\sF\cap\sT\neq \emptyset]\\
        \leq &2\sum_{\substack{\sF\in \mathsf{CC}\\ i, j\in \sF}} \prod_{e\in \sF} (2w_e+w_e^2) = 2\sum_{\substack{\sF\in \mathsf{CC}\\ i, j\in \sF}} \prod_{e\in \sF} \big(e^{4|z|\|\vh_e\|}-1\big)
    \end{align}
    where we note that the clusters $\sF, \sT$ from $\sA, \sB$ must be connected to define a non-trivial contribution to the commutator. We then group the sum by their union $\sF\cup \sT$ and observe that each edge appears in either $\sF, \sT$ or in both. Finally, we sum over all $j\in [n]\setminus \{i\}:$
    \begin{align}
        \sum_{j\neq i}  \mathsf{corr}_\beta(i, j) &\leq 2\sum_j \sum_{\substack{\sF\in \mathsf{CC}\\ i, j\in \sF}} \prod_{e\in \sF} \big(e^{4|\beta|\|\vh_e\|}-1\big)\\
        &\leq 2\sum_{\substack{\sF\in \mathsf{CC}\\ i\in \sF}} |\mathsf{supp}(\sF)|\cdot \prod_{e\in \sF} \big(e^{4|\beta|\|\vh_e\|}-1\big)
    \end{align}
    where in the last line we observe that each $j\in |\mathsf{supp}(\sF)|$ contributes the same factor to the sum.
\end{proof}

\subsubsection{Proof of \cref{lemma:cluster}}

To conclude this subsection, we present a proof of the cluster expansion. The proof is based on the combinatorial counting arguments in \cite{Neto2004} and simplified treatment from \cite[Lemma 2]{Mann_2021}. We also note a passing resemblance to the presentation in \cite{chen2021operator} for real-time evolution based on the  Schwinger-Karplus identity which inspired this inquiry.

\begin{proof}

    [Proof of \cref{lemma:cluster}]
    We denote $\mathsf{adj}_\vH[\vX] = [\vH, \vX]$ and $\mathsf{adj}_e[\vX]:=[\vh_e, \vX]$. The complex-time evolution admits the series expansion:
    \begin{align}
        \vA(z) =  \sum_{m}\frac{(iz)^m}{m!}\cdot  \mathsf{adj}^m_\vH[\vA]  = \sum_{m}\frac{(iz)^m}{m!}\sum_{e_1, \cdots, e_m} \bigg(\prod_j^m \mathsf{adj}_{e_j}\bigg)[\vA],
    \end{align}
    in terms of ordered multi-sets of edges $\mathbf{e}:=(e_1, \cdots, e_m)$, where each $e_j\in \Gamma$. 
    We group these sequences by the subset of unique edges they touch $\sU(\mathbf{e}) := \{e_1, \cdots, e_m\}\subseteq \Gamma$. 
    
    Fixed such a subset $\sF$:
    \begin{align}
        \vA_\sF(z) := \sum_m \frac{(iz)^m}{m!}\sum_{\sU(\mathbf{e})=\sF} \prod_j^m \mathsf{adj}_{e_j}[\vA]
    \end{align}
    we note that if $\sF\notin \mathsf{CC}$ is not connected, then $\vA_\sF(z)=0$ (recall, $\vA$ is single-site).  Otherwise, following \cite[Lemma 6.10, Lemma 6.21]{Neto2004}, we can control the norm of $ \vA_\sF(z)$ by counting the number of such ordered sequences $\mathbf{e}$ satisfying $\sU(\mathbf{e})=\sF$, by enumerating over the multiplicity $m_e\geq 1$ of each edge $e\in \sF$:
    \begin{align}
        \|\vA_\sF(z)\|&\leq \sum_{\sU(\mathbf{e})=\sF} \frac{|2z|^{|\mathbf{e}|}}{|\mathbf{e}|!} \prod_{i}^{|\mathbf{e}|} \|\vh_{e_i}\| \\
        &\leq \sum_{\{m_e\geq 1: e\in \sF\}} \frac{|2z|^{\sum_{e\in \sF} m_e}}{(\sum_{e\in \sF} m_e)!} \bigg(\prod_{e\in \sF} \|\vh_e\|^{m_e}\bigg)\cdot \bigg(\begin{aligned} & \quad \text{\# sequences of length $\sum_em_e$} \\ &\quad\text{with multiplicities $\{m_e\}_{e\in \sF}$}\end{aligned}\quad \bigg)\\
        &\leq \sum_{\{m_e\geq 1: e\in \sF\}} \frac{|2z|^{\sum_{e\in \sF} m_e}}{(\sum_{e\in \sF} m_e)!} \bigg(\prod_{e\in \sF} \|\vh_e\|^{m_e}\bigg)\cdot \frac{(\sum_{e\in \sF} m_e)!}{\prod_{e\in \sF}m_e!} \\
        &\leq \prod_{e\in \sF} \sum_{m_e\geq 1} \frac{|2z\|\vh_e\||^{m_e}}{m_e!} \\
        &\leq \prod_{e\in \sF} \big(e^{2|z|\|\vh_e\|}-1\big).
    \end{align}
    In the first line, we simply applied the norm bound $\|\mathsf{adj}_e\|_{\infty\rightarrow \infty}\leq 2\|\vh_e\|$. In the second, we used \cref{fact:counting-ordered} to count the number of ordered sequences of total length $\sum_{e\in \sF} m_e$ with multiplicities $\{m_e\}_{e\in\sF}$. At this point the summation factorizes over the various edges.
    
\end{proof}

We make use of the following trivial combinatorial counting fact. 

\begin{fact}[Multinomial coefficients]\label{fact:counting-ordered}
    The number of ways to partition a set of $n$ distinct elements into $k$ subsets where the number of elements per subset is $n_1, n_2\cdots n_k$ is
    \begin{equation}
        \frac{n!}{\prod_i^k n_i!}.
    \end{equation}
\end{fact}

\subsection{Sufficient conditions for the convergence of the cluster expansion}
\label{section:cte-conv}

In this section, we consider the special case of local Hamiltonians as defined in \cref{def:all-to-all}. We remind the reader:
\begin{enumerate}
    \item The \textit{locality} $\sk$ of $\vH$ is the maximum support size $|\mathsf{supp}(\vh_e)|\leq \mathsf{k}.$ 
    \item The \textit{degree} $\sd$ is the number of qudits any other qudit interacts with 
        \begin{equation}
        u\in [n]:\quad     \big|\{v:\exists e\in \Gamma \text{ s.t. } \{u, v\}\subseteq e\}\big|\leq \sd
        \end{equation}
        \item  The \textit{interaction strength} $\sJ$ is the sum over terms coupling a pair of qudits:
        \begin{equation}
            u\neq v: \sum_{e\supseteq\{u,v\}} \|\vh_e\|\leq \sJ
        \end{equation}
\end{enumerate}

  \begin{lemma}\label{lemma:conv-cluster}
      In the context of \cref{lemma:cluster}, let $\vH$ be any $\sk$-local Hamiltonian satisfying \cref{def:all-to-all}. Then, there exists a universal constant $\mathsf{c}$ s.t. $\forall |z|<z_\mathsf{c}:=\mathsf{c}\cdot [k\sd\sJ]^{-1}$,
      \begin{equation}
          \sum_{\substack{\sF\in\mathsf{CC}\\ i\in\sF}} |\sF|\cdot \prod_{e\in \sF} \big(e^{4|z|\|\vh_e\|}-1\big) \leq 100\cdot \sk\cdot \frac{|z|}{z_\mathsf{c}}
      \end{equation}
  \end{lemma}

  \begin{remark}
      We note that this entails bounds on $\mathsf{conv}_\beta$ via \cref{cor:cte-conv} and $\mathsf{corr}_\beta$ via \cref{cor:total-corr}.
  \end{remark}

  We note there is a long history of sufficient conditions for convergence of cluster expansions. Here we use a basic tree counting argument \cite{Simon1993LatticeGases, Seiler1982GaugeTheories}; more sophisticated recursive arguments are possible \`a la Kotecky-Preiss \cite{KoteckyPreiss1986,MIRACLESOLE2000244}.  
  
\begin{proof}
    Observe $|\mathsf{supp}(\sF)|\leq \sk\cdot |\sF|$, and thus it suffices to evaluate the sum conditioned on $|\sF|=m:$
    \begin{align}
        \sum_{\substack{\sF\in\mathsf{CC}\\ i\in\sF, |\sF|=m}} \prod_{e\in \sF} \big(e^{4|z|\|\vh_e\|}-1\big), \quad \text{and let}\quad w_e=e^{4|z|\|\vh_e\|}-1.
    \end{align}
    For this purpose, we overcount the sets $\sF\in\mathsf{CC}$ such that $i\in\sF$ and $|\sF| = m$ using a tree counting argument.

    Let $\mathbb{T}$ be the set of unlabeled trees on $m+1$ vertices. For each $T\in \mathbb{T}$, we label each vertex $v\in [m]$ by a Hamiltonian term $\sigma_T:[m]\rightarrow \Gamma$ and (in abuse of notation) label the root by the site $i\in [n]$. We refer to $\sigma_T$ as \textsf{valid} if on each edge $(u, v)\in T$, $\mathsf{supp}(\sigma_T(u))\cap \mathsf{supp}(\sigma_T(v))\neq \emptyset.$ Then:
\begin{align}
     \sum_{\substack{\sF\in\mathsf{CC}\\ i\in\sF, |\sF|=m}} \prod_{e\in \sF} w_e\leq \sum_{T\in \mathbb{T}} \sum_{\textsf{valid }\sigma_T} \prod_{j}^m w_{\sigma_T(j)}
\end{align}
    Let us now focus attention to a single $T\in \mathbb{T}$. Suppose WLOG the $m$th vertex is a leaf. Suppose we condition on a partial assignment $\sigma_{T\setminus m}$, and let us denote as $e_p\in \Gamma$ the assignment to the parent of this leaf. Then,
\begin{align}
    \sum_{\textsf{valid }\sigma_T} \prod_{j}^m w_{\sigma_T(j)} &=  \sum_{\textsf{valid }\sigma_{T\setminus m}} \prod_{j}^{m-1} w_{\sigma_T(j)} \cdot \sum_{\substack{\sigma_m \in \Gamma \\ \text{s.t. } \sigma_T \textsf{ valid}}} w_{\sigma_m} \\
    &\leq\sum_{\textsf{valid }\sigma_{T\setminus m}} \prod_{j}^{m-1}w_{\sigma_T(j)}\cdot \sum_{e:e\cap e_p\neq \emptyset}w_e \leq \Delta\cdot \sum_{\textsf{valid }\sigma_{T\setminus m}} \prod_{j}^{m-1}w_{\sigma_T(j)} \leq \Delta^{m} 
\end{align}   
where we overcount the valid assignments to the leaf by simply summing over all edges adjacent to $e_p$. In the above, we defined the parameter $\Delta\leq \sk\sd \big(e^{4|z|\sJ}-1\big)$ (\cref{lemma:istrength-param}). Using \cref{fact:tree-count} we have $|\mathbb{T}|\leq 4^{m+1}$, which enables us to conclude
\begin{align}
     \sum_{\substack{\sF\in\mathsf{CC}\\ i\in\sF}} |\sF|\cdot \prod_{e\in \sF} \big(e^{4|z|\|\vh_e\|}-1\big) \leq 4\sk \cdot \sum_m m\cdot (4\Delta)^m \leq \frac{16\sk \Delta}{(1-4\Delta)^2}.
\end{align}
To adequate to the lemma statement, we simply pick a suitable constant $\mathsf{c}$ such that $\sk\sd(\exp[4|z|\sJ]-1)\leq 8\sk\sd|z|\sJ = 8|z|/(\mathsf{c} z_{\mathsf{c}})$.

\end{proof}

\subsubsection{Deferred proofs}

The following interaction strength calculation is akin to \cite[Eq. 6.33]{Neto2004}.

\begin{lemma}\label{lemma:istrength-param}
    Consider the following interaction strength parameter:
    \begin{equation}
        \max_{e\in \Gamma}\sum_{e'\cap e\neq \emptyset} e^{4|z|\|\vh_{e'}\|}-1 \leq \sk\sd \big(e^{4|z|\sJ}-1\big).
    \end{equation}
\end{lemma}
\begin{proof}
    With $w_e=e^{4|z|\|\vh_e\|}-1$, we have:
    \begin{align}
        \max_{e\in \Gamma}\sum_{e'\cap e\neq \emptyset} w_{e'} \leq \sk\cdot \max_{u\in [n]}\sum_{e\ni u} w_{e}  &\leq \sk\sd \max_{u\neq v} \sum_{e\ni u, v}w_e 
        \\&\leq \sk\sd \max_{u\neq v} \sum_{e\ni u, v}e^{4|z|\|\vh_e\|}-1 \\&\leq \sk\sd\max_{u\neq v} \bigg(\exp\big[4|z|\sum_{e\ni u, v}\|\vh_e\|\big] -1\bigg) \\
        &\leq \sk\sd \big(e^{4|z|\sJ}-1\big),
    \end{align}
    where we used $x_1, \cdots x_k\geq 0\Rightarrow \sum_i(e^{x_i}-1)\leq e^{\sum x_i}-1$.
\end{proof}

\begin{fact}[{\cite[Theorem 1.5.1]{stanley2015}}]\label{fact:tree-count}
    The number of unlabeled trees $\mathbb{T}_m$ on $m$ vertices is $|\mathbb{T}_m|\leq 4^m$.
\end{fact}

\section{Dobrushin conditions for the pseudo-Lindbladian}
\label{section:dobrushin}

The ``Dobrushin condition'' and related ``Dobrushin-Shlosman conditions'' were developed as criteria on classical Gibbs measures to ensure a property known as \textit{uniqueness} of the measure in the thermodynamic limit \cite{dobrushin1985completely, dobrushin1985constructive, dobrushin1987completely}. A prolific subsequent line of work studied the implications of these conditions to rapid-mixing of spin systems \cite{HolleyStroock1987,HolleyStroock1989UniformL2,Aizenman1987,stroock1992equivalence,stroock1992logarithmic}.  

Recently, \cite{rouze2024optimal, bakshi2025dobrushinconditionquantummarkov} studied extensions of these conditions to non-commuting quantum spin systems, in order to prove rapid-mixing of such systems at high temperatures. Following their techniques, here we similarly develop sufficient conditions for the ``pseudo-Lindbladian'' $\CK$ of \eqref{eq:CK} to admit a system-size independent spectral gap. We shall see that the structure of $\CK$ 1) significantly simplifies the presentation, due to its explicit locality; however, 2) at the cost of only obtaining fast-mixing (not to be confused with rapid-mixing), since $\CK$ does not generate a CP map. 

We organize this section as follows. 

\begin{enumerate}
    \item We begin in \cref{section:osc} by introducing a modest variant of the oscillator norm, which we coin the KMS oscillator norm, and prove that its contraction entails a Poincar\'e inequality.
    \item In \cref{section:dob-conditions}, we establish a Dobrushin condition for the generator $\CK$ \cite{BC26} which implies contraction of the KMS-oscillator-norm, and in turn, a spectral gap for $\CK$.

    \item In \cref{section:dob-from-LR} we prove that said Dobrushin conditions reduce cleanly to properties of the complex-time evolution of local Pauli operators, previously derived in \cref{section:quasi-locality}.

\end{enumerate}

We then conclude this section with \cref{section:K-is-gapped}, where we put all these ingredients together to prove the main result of this section: 

\begin{theorem}[$\CK$ is gapped]\label{theorem:K-is-gapped}
    Let $\vH$ be a $\sk$-local Hamiltonian of local dimension $2^\sq$, degree $\sd$, pairwise strength $\sJ$, as in \cref{def:all-to-all}. Then, there exist  universal constants $\mathsf{c, \gamma}$ such that for any $\beta< [\mathsf{c}\sk^2\sJ\sd]^{-1}$, the generator $\CK$ \eqref{eq:CK} defined by all single-site Pauli jumps satisfies $\mathsf{gap}(\CK)\geq \gamma$.
\end{theorem}

\subsection{The KMS oscillator norm}
\label{section:osc}


We dedicate this subsection to a discussion on a modification to the oscillator norm, which we coin the KMS oscillator norm.  

\begin{definition}
    [KMS oscillator norm] For any operator $\vX$ and full rank density matrix $\vrho$, we refer to $\triple{\vX}_{\vrho}$ as the KMS-oscillator-norm of $\vX$, defined by:
    \begin{equation}
        \triple{\vX}_{\vrho}:= \sum_{i\in [n]} \|\delta_i[\vX]\|_{\vrho}, \quad \text{with}\quad \delta_i[\vX]:= \vX-\frac{1}{2^\sq}\tr_i[\vX].
    \end{equation}
\end{definition}

 We note that $-\delta_i\propto \CK_i^{0}$ is proportional to the generator of the depolarizing channel. The following lemma proves that if any (pseudo) Lindbladian $\CF$ generates an exponential decay in KMS oscillator norm, then it must be gapped. 

 \begin{lemma}[KMS-oscillator-norm decay implies a Poincar\'e inequality]\label{lemma:osc-to-gap}
     Let $\CF:\CB(\calH)\rightarrow \CB(\calH)$ be primitive and $\vrho$-detailed-balanced for a full rank state $\vrho$. Then, if there exists $\alpha>0$ such that
     \begin{equation}
       \forall\vX:\quad   \triple{e^{t\CF^\dagger}[\vX]}_{\vrho} \leq \triple{\vX}_{\vrho}\cdot e^{-\alpha\cdot t} \quad \Rightarrow \quad   \forall \vX:\quad \frac{\langle \vX, -\CF^\dagger[\vX]\rangle_{\vrho}}{\|\vX - \tr[\vrho \vX]\|_{\vrho}^2} \geq \alpha.
     \end{equation}
 \end{lemma}
 
 \begin{proof}
    Since $\vrho$ is full rank and $\CF$ is detailed-balanced, $\CF^\dagger$ admits a KMS orthonormal set of eigenoperators $\{\vB_i\}_i$:
    \begin{equation}
        \CF^\dagger[\vB_i]=-\lambda_i \vB_i, \quad i\neq j:\quad \langle \vB_i, \vB_j\rangle_{\vrho} = 0
    \end{equation}
    where $\lambda_i\in \BR$ and without loss of generality, $\vB_0=\vI$, $\lambda_0=0$ since it is primitive. Then, for each $i> 0:$ 
    \begin{align}
       e^{-\lambda_i t}\cdot  \triple{\vB_i}_{\vrho} = \triple{e^{t\CF}[\vB_i]}_{\vrho} \leq \triple{\vB_i}_{\vrho}\cdot e^{-\alpha\cdot t} \Rightarrow \lambda_i \geq \alpha.
    \end{align}
    where the first inequality is the assumption of KMS oscillator norm decay, and the implication $(\Rightarrow)$ follows since $\triple{\vB_i}_{\vrho}$ is non-zero if $\vB_i\neq \vI$.
 \end{proof}

\subsection{Dobrushin conditions for \texorpdfstring{$\CK$}{K}}
\label{section:dob-conditions}

We proceed by providing a set of sufficient conditions on the generator $\CK$, to ensure the decay in the KMS oscillator norm. Our approach is based on \cite{rouze2024optimal, MajewskiOlkiewiczZegarlinski1998}, with minor  modifications to address the fact that the generator $\CK$ is not completely positive. We remind the reader:

\begin{observation}
    For the purposes of showing the decay in the oscillator norm, what makes $\CK$ special is the following identity regarding the kernel of a single local term $\CK_i$:
\begin{equation}\label{eq:dis-spread}
    \CK_i[\vX] = \CK_i\circ \delta_i[\vX], \quad \delta_i[\vX]:= \vX-\frac{1}{2^\sq}\tr_i[\vX]
\end{equation}
\end{observation}

Indeed, $\delta_i$ has the effect of removing the identity term on qubit $i$, which is stationary under $\CK_i$. This suggests that if $\vX$ is already ``in thermal equilibrium'' on $i$, then updates on $i$ cannot further ``spread disagreements'' outside $i$. We emphasize this is categorically not true for the \cite{chen2023efficient} Lindbladian.

\begin{assumption}\label{assumption:dobrushin}
    We introduce the following assumptions on the locality of complex-time evolution and the pseudo-Lindbladian generator $\CK=\sum_i\CK_i$.
    \begin{enumerate}
        \item \textbf{Locality in temperature.} For each $i\in [n]$, there exist constants $\lambda, \eta_i\geq 0$ s.t.:
    \begin{equation}
        \|\CK_i+\lambda\cdot \delta_i\|_{\vrho} \leq \eta_i. \label{eq:t-approx} 
    \end{equation}
    \item \textbf{Pairwise influence.} For each $i\neq j\in [n]$, there exists a constant $\kappa_{ij}\geq 0$ such that:
    \begin{equation}
        \|[\delta_j, \CK_i]\|_{\vrho} \leq \kappa_{ij}.\label{eq:mf-approx}
    \end{equation}
    \end{enumerate}
\end{assumption}

We will prove shortly that these constants $\eta, \kappa$ exist, and scale linearly in $\beta$ at sufficiently small $\beta$, for our systems of interest. A trivial calculation below shows $\lambda=2\cdot 4^\sq.$

\begin{remark}
    \eqref{eq:t-approx} simply imposes that $\CK_i$, as a function of $\beta$, is a perturbation of the $\beta=0$ case. In turn, \eqref{eq:mf-approx} resembles a mean-field approximation, in that ``on average'' updates on $i\in [n]$ shouldn't affect all the other sites $j\in [n].$
\end{remark}

We are now in a position to state the main result of this section. 

\begin{lemma}
    [A Dobrushin condition for $\CK$]\label{lemma:d-conditions}
    Assume that the generator $\CK$ \eqref{eq:CK} satisfies \cref{assumption:dobrushin}, and furthermore:
    \begin{equation}
        \gamma:=  \max_i\bigg[\|\delta_i\|_{\vrho}\cdot \eta_i + \sum_{j\neq i}  \kappa_{ij}\bigg] < \lambda
    \end{equation}
    Then for any $\vX$, the KMS oscillator norm of $\vX_t:=e^{t\CK}[\vX]$ decays as $ \triple{\vX_t}_{\vrho} \leq e^{-t(\lambda-\gamma)}\cdot \triple{\vX}_{\vrho}.$
\end{lemma}

A consequence of the discussion on the KMS oscillator norm (\cref{lemma:osc-to-gap}) from the previous section, is then that the Dobrushin condition above implies a spectral gap for $\CK$.

\begin{remark}
    In the above, $\|\delta_i\|_{\vrho}$ is the induced superoperator $\vrho$ norm of $\delta_i$. As we discuss shortly, the convergence of complex-time evolution will imply it is bounded by a constant. 
\end{remark}

\begin{proof}[Proof of \cref{lemma:d-conditions}]  For any fixed $i\in [n]$, the time-derivative of $\vX_t$ satisfies the expansion:
\begin{align}
    \frac{\rd}{\rd t}\vX_t &= \sum_j\CK_j[\vX_t]= -\lambda\delta_i[\vX_t] + (\CK_i+\lambda\delta_i)[\vX_t] + \sum_{j\neq i} \CK_j[\vX_t]
\end{align}
where we recall that $-\delta_i$ is the generator for the depolarizing channel. 
Following \cite{rouze2024optimal}, we claim (and justify shortly) that $\delta_i[\vX_t]$ satisfies the bound:
\begin{align}
    \|\delta_i[\vX_t]\|_{\vrho} \leq e^{-\lambda t}\cdot \|\delta_i[\vX]\|_{\vrho} + \int_0^t e^{\lambda(z-t)}\cdot \bigg(\|\delta_i\circ (\CK_i+\lambda\delta_i)[\vX_z]\|_{\vrho}  + \sum_{j\neq i}\|[\delta_i, \CK_j][\vX_z]\|_{\vrho} \bigg)\rd z.\label{eq:rfa-eq}
\end{align}
We can now apply the influence assumptions in the lemma statement. Assumption \eqref{eq:t-approx} and the kernel condition \eqref{eq:dis-spread}, applied to $\vY=\calP_z[\vX]$, implies a bound on the first term in the RHS above. 
\begin{align}
    \|\delta_i\circ (\CK_i+\lambda\delta_i)[\vY]\|_{\vrho}& \leq  \|\delta_i\|_{\vrho}\cdot \| (\CK_i+\lambda\delta_i)[\vY]\|_{\vrho} \leq \eta_i \cdot  \|\delta_i\|_{\vrho}\cdot \|\delta_i[\vY]\|_{\vrho}
\end{align}
where we used that $\delta_i^2=\delta_i$, and $\|\delta_i\|_{\vrho}$ is the induced superoperator $\vrho$ norm. In turn, for the second term we apply assumption \eqref{eq:mf-approx} and the kernel condition \eqref{eq:dis-spread}:
\begin{align}
    \|[\delta_i, \CK_j] [\vY]\|_{\vrho} &\leq \kappa_{ji}\cdot \|\delta_j [\vY]\|_{\vrho}
\end{align}
  Summing \eqref{eq:rfa-eq} over $i\in [n]$ then exposes the desired KMS-oscillator norm:
\begin{align}
    \triple{\vX_t}&\leq e^{-\lambda t}\triple{\vX} + \max_i\bigg(\eta_i\|\delta_i\|_{\vrho}+\sum_{j\neq i} \kappa_{ij}\bigg) \int^t_0 e^{\lambda(z-t)} \cdot \triple{\vX_z}\rd z 
\end{align}
which results in the advertised bound. 

It remains to then justify \eqref{eq:rfa-eq}. Let us denote $\calP_t = e^{t\CK}$, and $\calP_t^{(i)} = e^{t(\CK-\CK_i)}$ to be the generator with the updates on $i$ removed. The following calculation, \cite[Equation (A5)]{rouze2024optimal} controls the deviation between the two time-evolutions via:
\begin{align}\label{eq:rfa-a5}
    \frac{\rd}{\rd z}\bigg(e^{\lambda z}\cdot  \calP^{(i)}_{t-z}\circ \delta_i\circ \calP_z[\vX]\bigg) = e^{\lambda z}\cdot \calP_{t-z}^{(i)}\circ \bigg(\sum_{j\neq i}[\delta_i, \CK_j]\circ \calP_z[\vX]+\delta_i\circ (\CK_i+
    \lambda\delta_i)\circ \calP_z[\vX]\bigg).
\end{align}
Integrating on both sides and using the contractivity of $\calP_t, \calP_t^{(i)}$ in the $\vrho$ norm then gives the desired statement. 
\end{proof}

\begin{remark}
    We remark the time-evolution $\calP_t$ in \eqref{eq:rfa-a5} need not be contractive in the $\infty$ norm, since $\CK$ is not CP. Nevertheless, the negative semi-definite spectra implies it is contractive in (KMS) $\vrho$ norm (\cref{lemma:K-properties}), which is why we required the modification to the oscillator norm. 
\end{remark}

\subsection{Dobrushin conditions from the quasi-locality of complex-time evolution}
\label{section:dob-from-LR}

We dedicate this subsection to simplifying the Dobrushin conditions above (the approximations in time and space for the generator $\CK$), to certain operator-growth bounds. As we shall see, each of the parameters in \cref{lemma:d-conditions} admits a clean reduction to some property of the complex-time evolution of local Pauli operators. 

We refer the reader back to \cref{def:cte} and \cref{def:pairwise-corr} for the definitions of the complex-time evolution and pairwise-correlations constants $ \mathsf{conv}_\beta, \mathsf{corr}_\beta (i, j)$. Equipped with these definitions, we can now simplify each of the parameters in the Dobrushin condition. The proofs are deferred to \cref{section:dob-def-proof} below. We begin with the simplest of the three: the superoperator norm of the partial trace. 

\begin{lemma}\label{lemma:tr-induced-norm}
    The induced superoperator $\vrho$ norm of $\delta_i$ is bounded by $\|\delta_i\|_{\vrho}\leq 2(1+\mathsf{conv}_{\beta/4})^2$. 
\end{lemma}

Next, we turn to the locality in temperature condition.

\begin{lemma}
    [Locality in temperature]
    \label{lemma:Kt-approx}
    For each site $i\in [n]$, the generator $\CK_i$ at inverse-temperature $\beta$ satisfies the approximation:
    \begin{equation}
         \|\CK_i^{(\beta)} -\CK_i^{(0)}\|_{\vrho} \leq 2^{2\sq+3}\cdot \mathsf{conv}_{\beta/4}\cdot (1+\mathsf{conv}_{\beta/4})
    \end{equation}
\end{lemma}

Finally, we turn to the mean-field approximation.

\begin{lemma}
    [Mean-field approximation for $\CK$]
    \label{lemma:Kmf-approx}
    For any pair of sites $i, j\in [n]$, the generator $\CK_i$ at inverse-temperature $\beta$ satisfies the approximation
    \begin{equation}
          \|[\delta_j, \CK_i]\|_{\vrho} \leq  2^{2\sq+2}\cdot (1+\mathsf{conv}_{\beta/4})^2\cdot \mathsf{corr}_{\beta/4} (i, j)
    \end{equation}
    
\end{lemma}

Put together, we can now completely recast the Dobrushin condition of \cref{lemma:d-conditions} in terms of operator-growth conditions. 

\begin{corollary}\label{corr:cte-dobrushin}
    Assume that the parameters $\mathsf{conv}_{\beta}, \mathsf{corr}_{\beta}$ from \cref{def:cte,def:pairwise-corr} satisfy the inequality
    \begin{equation}
        2^4\cdot \mathsf{conv}_{\beta/4}+2^2\cdot \mathsf{corr}_{\beta/4} <1.
    \end{equation}
    Then, $\CK$ is gapped. 
\end{corollary}
\begin{proof}
    Substitution of the bounds above in \cref{lemma:d-conditions} gives:
    \begin{equation}
        2^3\cdot \mathsf{conv}_{\beta/4}\cdot (1+\mathsf{conv}_{\beta/4})^3+2\cdot (1+\mathsf{conv}_{\beta/4})^2\cdot \mathsf{corr}_{\beta/4} <1.
    \end{equation}
    The observation that this bound is trivial unless $ \mathsf{conv}_{\beta/4}\leq 8^{-1}$ and arithmetic simplifications gives the desired result. 
\end{proof}

\subsubsection{Deferred proofs}
\label{section:dob-def-proof}

\begin{proof}

    [Proof of \cref{lemma:tr-induced-norm}]
    We first note that $\delta_i$ is proportional to the generator of the depolarizing semigroup:
    \begin{align}
        \delta_i[\vX]=\vX-\frac{1}{2^\sq}\tr_i[\vX] = \frac{1}{2\cdot 4^\sq}\sum_{a\in \CS_i} [\vA^a, [\vA^a, \vX]]
    \end{align}
    Consequently, the induced superoperator $\vrho$ norm satisfies:
    \begin{equation}
        \|\delta_i\|_{\vrho} \leq \frac{1}{2\cdot 4^\sq}\cdot 4^\sq \cdot 4\cdot \max_{a \in \CS^1_{[n]}}\sup_{z\in [-\beta, \beta]}\|e^{z \vH/4}\vA^a e^{- z\vH/4}\|^2 \leq 2(1+\mathsf{conv}_{\beta/4})^2,
    \end{equation}
    by application of KMS H\"older's inequality \cref{lem:kms-holder}. 

\end{proof}

\begin{proof}

    [Proof of \cref{lemma:Kt-approx}]
    Let us restrict our attention to a single site jump operator $a\in \CS_i$.
    From the KMS-H\"older's inequality \cref{lem:kms-holder} we have, assuming $\|\vX\|_{\vrho}=1$ and using the explicit form of $\CK_a$ \eqref{eq:CK}:
    \begin{align}
        \|(\CK_a^{(\beta)} -\CK_a^{(0)})[\vX]\|_{\vrho} \leq 4\cdot \sup_{z\in [-\beta, \beta]} \| e^{z\vH/4}(\vA^a(iz/2)-\vA^a)e^{-z\vH/4}\|\cdot \sup_{z\in [-\beta, \beta]} \| \vA^a(iz/4)\|\label{eq:kt-approx-holder}
    \end{align}
    The second term on the RHS above is precisely the convergence of complex time evolution. 
    \begin{equation}
        \sup_{z\in [-\beta, \beta]} \| \vA^a(iz/4)\| \leq 1+\mathsf{conv}_{\beta/4}
    \end{equation}
    In turn, the first term in \eqref{eq:kt-approx-holder} follows by the triangle inequality, and again invoking convergence of complex-time evolution. 
    \begin{equation}
         \|e^{z\vH/4}(\vA^a(iz/2)-\vA^a)e^{-z\vH/4}\| \leq 2\cdot \sup_{z\in [-\beta, \beta]} \|\vA^a(iz/4)-\vA^a\| \leq 2\cdot \mathsf{conv}_{\beta/4}
    \end{equation}
    Incorporating the sum over the $4^\sq$ jump operators $a\in \CS_i$ on $i$ then gives the desired bound. 
\end{proof}

\begin{proof}

    [Proof of \cref{lemma:Kmf-approx}] We require an expansion of 
    \begin{equation}
         \delta_j[\vX]=\vX-\frac{1}{2^\sq}\tr_j[\vX] = \frac{1}{ 4^\sq}\sum_{b\in \CS_j} \vX-\vB^b\vX\vB^b.
    \end{equation}
    For simplicity, let us focus our attention on a single local jump operator in $\CK_i=\sum_{a\in \CS_i} \CK_a$. Under the expansion above, we can simplify the desired commutator norm as:
    \begin{align}
        \|[\delta_j, \CK_a][\vX]\|_{\vrho} \leq 4^{-\sq}\sum_{b\in \CS_j} \|\vB^b\CK_a[\vX ](\vB^b)^\dagger -\CK_a[\vB^b\vX (\vB^b)^\dagger]\|_{\vrho}
    \end{align}
    We now fix our attention on a single term $b\in \CS_j$ as well. By means of the explicit formula for $\CK_a$ \eqref{eq:CK}, and KMS H\"older's inequality:
    \begin{align}
        &\|\vB^b\CK_a[\vX ](\vB^b)^\dagger -\CK_a[\vB^b\vX (\vB^b)^\dagger]\|_{\vrho}\\
        = &\bigg\| \vB^b[\vA^a, \vX][\vA^a(iz/2)^{\dagger}, \vB^{b \dagger}] - [\vB^b, \vA^{a}(-iz/2)^{\dagger}][\vA^a, \vX]\vB^b\bigg\|_{\vrho}\\
        \leq  &4\cdot (1+\mathsf{conv}_{\beta/4})^2 \cdot \max_{b\in \CS_j} \sup_{z\in [-\beta, \beta]}   \|e^{z\vH/4}[\vA^{a\dagger}(iz/2), \vB^{b \dagger}]e^{-z\vH/4}\| \\
        = &4\cdot (1+\mathsf{conv}_{\beta/4})^2 \cdot\sup_{\substack{z\in [-\beta, \beta]\\ b\in \CS_j}} \|[\vA^{a\dagger}(iz/4), \vB^b(iz/4)]\| \\\leq & 4\cdot (1+\mathsf{conv}_{\beta/4})^2\cdot \mathsf{corr}_{\beta/4} (i, j)
    \end{align}
Incorporating the sums over $a\in \CS_i, b\in \CS_j$ then gives the claimed result.
    
\end{proof}

\subsection{\texorpdfstring{$\CK$}{K} is gapped at high temperatures}
\label{section:K-is-gapped}

We dedicate this subsection to collecting the ingredients from the previous subsections in order to prove \cref{theorem:K-is-gapped}.

\begin{proof}

    [Proof of \cref{theorem:K-is-gapped}] We collect the statements from the sections above in the following sequence:

    \begin{enumerate}
        \item From the convergence of the cluster expansion \cref{lemma:conv-cluster}, we have that so long as $\beta< \mathsf{c} [\sk^2\sJ\sd]^{-1}$ for a suitable universal constant $\mathsf{c}$, the convergence of complex-time evolution constants \cref{def:cte}, \cref{def:pairwise-corr} satisfy:
        \begin{equation}
            \mathsf{conv}_{\beta/4}, \mathsf{corr}_{\beta/4} \leq \frac{1}{100}.
        \end{equation}
        \item From the relationship between convergence of complex-time evolution and the quasi-locality of $\CK$ (\cref{lemma:Kt-approx}, \cref{lemma:Kmf-approx}), the local generators $\CK_i$ satisfy the guarantees of the Dobrushin condition for $\CK$ (\cref{assumption:dobrushin}). 
        
        \item In particular, the conditions of \cref{corr:cte-dobrushin}
        \begin{equation}
            2^4\cdot \mathsf{conv}_{\beta/4}+2^2\cdot \mathsf{corr}_{\beta/4} <1.
        \end{equation}
        are satisfied, and thus $\CK$ admits contraction in oscillator norm.  

        \item From \cref{lemma:osc-to-gap}, the decay in oscillator norm implies a spectral gap. 
    \end{enumerate}
\end{proof}

\section{Dirichlet form comparison arguments}
\label{section:d-form-comparison}

We dedicate this subsection to relating the spectral gaps of the  generators $\CL:=\CL_{\mathsf{CKG}}$ \cite{chen2023efficient} and $ \CK$, in the high-temperature regime. We refer the reader to \cref{section:prelim} for a description of $\CL$ and the operator Fourier transform. The main result of this subsection is the following theorem:

\begin{theorem}
    [Comparing $\CL_{\mathsf{CKG}}$ and $\CK$]
    \label{theorem:K-to-ckg} Let $\vH$ be an arbitrary $\sk$-local Hamiltonian of bounded degree $\sd$ and maximum pairwise interaction strength $\sJ$ (\cref{def:all-to-all}). Then, there exists a pair of universal constants $\mathsf{c}_1\geq 1, \mathsf{c}_2> 0$ such that for any inverse-temperature $\beta < [\sJ\sd \mathsf{c}_1^{\sq\sk}]^{-1}$:

    \begin{equation}
      \forall \vO:\quad   \langle \vO,-\CL^\dagger[\vO]\rangle_{\vrho} \geq \mathsf{c}_2^{\sq\sk}\cdot  \langle \vO,-\CK[\vO]\rangle_{\vrho}
    \end{equation}
    with $\CK, \CL$ the generators defined by single-site Pauli jumps and the energy filter $\sigma=\beta^{-1}$ in $\CL$.
\end{theorem}

As a corollary of \cref{theorem:K-is-gapped}, we have then that $\CL$ is gapped at sufficiently high temperatures in all-to-all systems, and thus conclude the proof of \cref{theorem:all-to-all-mixing}.

\begin{remark}
 The proof of \cref{theorem:K-to-ckg} is unfortunately not a direct reduction to the convergence of complex-time evolution. As such, we currently do not know how to establish fast-mixing directly from the convergence of the cluster expansion. Nevertheless, the comparison can be re-routed to hold at sufficiently high temperatures. 
\end{remark}

\subsection{The Dirichlet form of \texorpdfstring{\cite{chen2023efficient}}{[CKG23]}}

We require the following explicit formula for the Dirichlet form of $\CL$. 

\begin{lemma}[The Dirichlet Form of $\CL$ {~\cite[Lemma X.2]{chen2025GibbsMarkov}}]\label{lemma:Dirichlet}
    Fix a single Hermitian jump operator $\vA^a = \vA^{a, \dagger}$. The Dirichlet form for the local Lindbladian $\CL_a$ \cite{chen2023efficient} associated to the jump operator $\vA^a$ with frequency width $\sigma\in \mathbb{R}^+$, can be written as
\begin{align}
    \CE_{a}(\vO) &:= -\langle\vO,\CL_a^{\dagger}[\vO]\rangle_{\vrho}=  \iint_{-\infty}^{\infty} g(t)h(\omega) \cdot \norm{[\hat{\vA}^a(\omega,t),\vO]}^2_{\vrho} \cdot  \rd t \rd \omega.\label{eq:dirichlet_L}
\end{align}
where $\hat{\vA}^a(\omega,t):=e^{i\vH t}\hat{\vA}^a(\omega)e^{-i\vH t}$,  and $g(t) = \frac{1}{\beta\cosh(2\pi t/ \beta)}\ge 0.$ Further, the frequency filter $h(\omega)$ depends on the choice of transition rate $\gamma(\omega)$; here we focus on the Metropolis weight. 
\begin{align}
 h_\mathsf{M}(\omega) = e^{-\sigma^2\beta^2/8}e^{-|\omega|\beta/2}
\end{align}
\end{lemma}

Following the intuition in \cite{BC26}, the reason the Dirichlet form comparison in \cref{theorem:K-to-ckg} could even be possible is because the kernel of the local generator $\CK_a$ contains that of $\CL_a^\dagger$:
\begin{equation}
    \langle\vO,\CL_a^{\dagger}[\vO]\rangle_{\vrho}=0 \iff  \forall \omega, t:\quad [\hat{\vA}^a(\omega, t), \vO] = 0 \Rightarrow [\vA^a, \vO] = 0 \iff \langle \vO,\CK_a^{\dagger}[\vO]\rangle_{\vrho}=0
\end{equation}
while this suggests that the rate of change of observables under $\CL_a^\dagger$ should be ``faster'' than that under $\CK_a$, performing such a comparison is non-trivial: the challenge lies in controlling how elements outside of the kernel of $\CK_a$ may interfere.

\subsubsection{Proof of \cref{theorem:K-to-ckg}}

Although semantically quite similar to the analogous comparison in \cite{BC26}, their proof hinged on a shell decomposition for the OFT, which currently does not seem within reach for all-to-all quantum systems. Instead, here we employ a strategy where we sequentially  ``peel off''  the frequency and time-dependence in the integral expression in the Dirichlet form above (\cref{lemma:Dirichlet}), by relying on the convergence of complex-time evolution.

\begin{proof}

    [of \cref{theorem:K-to-ckg}] The proof follows by first relating the Dirichlet Form of $\CL$ to an auxiliary Dirichlet form $ \CC[\vO]$ where the frequency dependence is removed, which we refer to as the $g(t)-$Dirichlet form:
\begin{align}
        \langle \vO,-\CL^\dagger[\vO]\rangle_{\vrho} \geq c_1\cdot \CC[\vO], \quad
        \CC[\vO] =  \sum_{a\in \CS_{[n]}^1}  \int_{-\infty}^\infty  g(t)\cdot \big\|[\vA^a(t), \vO]\|_{\vrho}^2\rd t 
\end{align}
this is the content of \cref{lemma:ckg-to-gt} (below). Note that the only difference to \cref{lemma:Dirichlet} is that the integral over $\omega$ has disappeared. Subsequently, we relate the $g(t)-$Dirichlet form to that of $\CK$:
\begin{equation}
    \CC[\vO] \geq c_2\cdot \langle \vO, -\CK[\vO]\rangle_{\vrho}, \quad  \langle \vO, -\CK[\vO]\rangle_{\vrho}= \sum_{a\in \CS_{[n]}^1}  \|[\vA, \vO]\|_{\vrho}^2
\end{equation}
by leveraging \cref{gt-dirichlet} (also below). Put together then gives the desired theorem. 
\end{proof}

\subsection{Peeling off the time and frequency dependence}

It remains only to prove \cref{gt-dirichlet} and  \cref{lemma:ckg-to-gt}, presented in this order for simplicity:

\begin{lemma}[Peeling off the time dependence]\label{gt-dirichlet}

Let $\vH$ be a $\sk$-local Hamiltonian of bounded degree $\sd$ and maximum pairwise interaction strength $\sJ$ (\cref{def:all-to-all}). Assume the convergence of complex-time evolution of single-site Pauli operators with constant $\mathsf{conv}_\beta$ (\cref{def:cte}).  

Then the $g(t)$ Dirichlet form is bounded by that of $\CK$:
    \begin{equation}
      \sum_{a\in \CS_{[n]}^1}  \int_{-\infty}^\infty  g(t)\cdot \big\|[\vA^a(t), \vO]\|_{\vrho}^2\rd t  \geq \gamma\cdot  \sum_{a\in \CS_{[n]}^1}  \|[\vA, \vO]\|_{\vrho}^2
    \end{equation}
    where $\gamma>0$ is a constant dependent only on $\sq, \sk, \mathsf{conv}_\beta$ and $\beta \cdot \sJ\cdot \sd$.
\end{lemma}

\begin{proof}

    [of \cref{gt-dirichlet}] For conciseness, let us denote the single-site $g(t)-$Dirichlet form as:
    \begin{equation}
        \CC_i[\vO] = \sum_{a\in \CS_{i}^1} \int_{-\infty}^\infty g(t) \|[\vA(t), \vO]\|_{\vrho}^2 \rd t
    \end{equation}
    Let us next fix a single jump operator $\vA$ supported on site $i\in [n]$. By Duhamel's formula, one can expand its time evolution in terms of the time-evolution of commutators:
    \begin{equation}
        \vA(t) - \vA = i\int_0^t e^{i\vH s} [\vH, \vA]e^{-i\vH s} \rd s.
    \end{equation}
    The reverse triangle inequality then exposes the Dirichlet form of $\CK_a$ of a given operator $\vO$:
    \begin{align}\label{eq:rev-triangle-t}
        \|[\vA(t), \vO]\|_{\vrho} \geq \|[\vA, \vO]\|_{\vrho} - \int_0^t \bigg\| \big[ [\vH, \vA](s), \vO\big]\bigg\|_{\vrho} \rd s
    \end{align}
    To proceed, we expand the commutator $[\vH, \vA] = \sum_{e\ni i} [\vh_{e}, \vA]$ as a linear combination of $\sk$-local Pauli operators, each of bounded norm:    
    \begin{align}
        \bigg\| \big[ [\vH, \vA](s), \vO\big]\bigg\|_{\vrho} &\leq 2^{2\sq\sk}\cdot  \sum_{e\ni i}  \|\vh_{e}\| \cdot \sum_{\vP\in \CS_e} \|[\vP(s), \vO]\|_{\vrho}\label{eq:comm-comm-A-gt}
    \end{align}
    with $\vP\in \CS_e$ the sum over all $\sk$-site Pauli operators on $e\subseteq [n]$. To proceed, we leverage the convergence of complex-time evolution and KMS H\"older's inequality \cref{lem:kms-holder} to relate the commutators of $\sk$-local Paulis to single-site Paulis. Let us denote as $c_\beta:=\max_{\vB\in \CS^1_{[n]}}\|e^{\pm \beta \vH}\vB e^{\mp \beta\vH}\|\leq 1+\mathsf{conv}_\beta$. Then, 
    \begin{equation}
      \forall\vP\in \CS_e:\quad   \|[\vP(s), \vO]\|_{\vrho} \leq 2(2c_\beta)^{\sk-1}\cdot \max_{\vB\in \CS^1_{e}} \|[\vB(s), \vO]\|_{\vrho},
    \end{equation}
    which gives us a much simplified version of 
    \begin{align}
        \eqref{eq:comm-comm-A-gt} &\leq (2^{2\sq}c_\beta)^\sk\cdot \sum_{e\ni i}  \|\vh_{e}\| \cdot \max_{\vB\in \CS^1_{e}} \|[\vB(s), \vO]\|_{\vrho} \\
        &\leq (2^{2\sq}c_\beta)^\sk\cdot  \sum_{e\ni i}\|\vh_{e}\| \cdot \sum_{j\in e}   \sum_{\vB\in \CS^1_{j}} \|[\vB(s), \vO]\|_{\vrho}.\label{eq:multi-to-single-gt}
    \end{align}
 This allows us to return to \eqref{eq:rev-triangle-t}:
\begin{align}
     \|[\vA(t), \vO]\|_{\vrho}^2 \geq   \|[\vA, \vO]\|_{\vrho}^2 - 2(2^{2\sq}c_\beta)^\sk\cdot  \|[\vA, \vO]\|_{\vrho} \cdot  \sum_{e\ni i}\|\vh_{e}\| \cdot \sum_{j\in e}   \sum_{\vB\in \CS^1_{j}} \int_0^t    \|[\vB(s), \vO]\|_{\vrho} \rd s\label{eq:comm-comm-simp-gt}
    \end{align}
    To proceed, we note that it suffices to consider the second term on the RHS above. To do so, we proceed by integrating over $g(t)$ as in the desired expression, by first using the fact $g(t)$ is sub-multiplicative: $0\leq x\leq t:g(t)\leq \beta g(x)g(t-x)$, and then sequentially applying the Cauchy-Schwarz inequality:
    \begin{align}
      \sum_{\vB\in \CS^1_{j}}  \int_0^\infty g(t) \int_0^t   \|[\vB(s), \vO]\|_{\vrho} \rd s \rd t &\leq \beta\cdot  \int_0^\infty \rd z g(z)\times  \sum_{\vB\in \CS^1_{j}}\int_0^\infty g(s)   \|[\vB(s), \vO]\|_{\vrho} \rd s \\
       &\leq \beta\cdot 2^{\sq}\cdot \bigg[ \sum_{\vB\in \CS^1_{j}}\int_0^\infty g(s)   \|[\vB(s), \vO]\|_{\vrho}^2 \rd s\bigg]^{1/2}\\
       &= \beta \cdot 2^\sq\cdot \CC_j[\vO]^{1/2}\label{eq:submult-application-gt}
    \end{align}
    which then exposes the $g(t)$-Dirichlet form at site $j$. This enables us to finally sum over $i\in [n]$ and single-jump operators $\vA$. For conciseness, let us denote $\eta_1:=\beta 2^{\sq+1}(2^{2\sq}c_\beta)^\sk$:
    \begin{align}
        \sum_i \CC_i[\vO] &\geq  \sum_{\substack{i\in [n] \\ a\in \CS_{i}^1}}\|[\vA^a, \vO]\|_{\vrho}^2 -  \eta_1 \sum_{\substack{i\in [n] \\ a\in \CS_{i}^1}}\|[\vA^a, \vO]\|_{\vrho}\cdot \sum_{e\ni i}\|\vh_{e}\| \cdot \sum_{j\in e}\CC_j[\vO]^{1/2} \\
        &\geq  \sum_{\substack{i\in [n] \\ a\in \CS_{i}^1}}\|[\vA^a, \vO]\|_{\vrho}^2 -  \eta_1 \sum_{\substack{i\in [n] \\ a\in \CS_{i}^1}}\|[\vA^a, \vO]\|_{\vrho}\cdot \bigg( \alpha\cdot \CC_i[\vO]^{1/2} + \sJ\cdot \sum_{j\sim i} \CC_j[\vO]^{1/2} \bigg)\\
        &\geq \sum_{ a\in \CS_{[n]}^1}\|[\vA^a, \vO]\|_{\vrho}^2 - 2^\sq\eta_1 (\alpha+\sd\sJ)\bigg[\sum_{ a\in \CS_{[n]}^1}\|[\vA^a, \vO]\|_{\vrho}^2\bigg]^{1/2} \cdot \bigg(\sum_i \CC_i[\vO]\bigg)^{1/2}\label{eq:we-are-done} 
    \end{align}
    where we recall the definitions: $\alpha=\max_i\sum_{e\ni i} \|\vh_e\|, \sJ=\max_{i\neq j} \sum_{e\ni i, j}\|\vh_e\|$ and $\sd=\max_i\{j: \exists e \text{ s.t. }i, j\in e\}$, such that $\alpha\leq \sd\cdot \sJ$.

     Finally, the desired result follows from case division. Assuming, for the purposes of contradiction, that:
    \begin{equation}
        \sum_i \CC_i[\vO] \leq c_4^2\cdot \sum_{\substack{i\in [n] \\ a\in \CS_{i}^1}}\|[\vA, \vO]\|_{\vrho}^2 
    \end{equation}
    then \eqref{eq:we-are-done} implies
    \begin{equation}
         \sum_i \CC_i[\vO] \geq  \sum_{\substack{i\in [n] \\ a\in \CS_{i}^1}}\|[\vA, \vO]\|_{\vrho}^2\cdot \bigg(1-\beta \cdot 2^{2\sq+2}(2^{2\sq}c_\beta)^\sk\cdot c_4\cdot   \sd\cdot \sJ \bigg)
    \end{equation}
    appropriately choosing $c_4$ then gives the desired bound.     
\end{proof}

Next, we show how to relate the \cite{chen2023efficient} Dirichlet form to the $g(t)-$Dirichlet form.

\begin{lemma}
    [\cite{chen2023efficient} to $g(t)$] \label{lemma:ckg-to-gt}
    Let $\vH$ be a $\sk$-local Hamiltonian of bounded degree $\sd$ and maximum pairwise interaction strength $\sJ$ (\cref{def:all-to-all}). There exists a universal constant $\mathsf{c}>0$ such that for any $\beta < (\sJ\sd \mathsf{c}^{\sq\sk})^{-1}$:
     \begin{equation}
      \sum_{a\in \CS_{[n]}^1}  \iint_{-\infty}^\infty h_\sM(\omega) g(t)\cdot \big\|[\hat{\vA}^a(\omega, t), \vO]\|_{\vrho}^2\rd t \rd \omega  \geq \frac{1}{4} \cdot \int_{-\infty}^\infty g(t) \sum_{a\in \CS_{[n]}^1}  \|[\vA(t), \vO]\|_{\vrho}^2 \rd t
    \end{equation}
    where $\hat{\vA}^a(\omega)$ is the operator Fourier transform with energy filter $\sigma=\beta^{-1}$.
\end{lemma}

\begin{proof}

   We begin as in the proof of \cref{gt-dirichlet} by applying Duhamel's formula, now to the OFT:
    \begin{align}
         \hat{\vA}(\omega) -\vA \hat{f}_\sigma(\omega) &= \frac{1}{\sqrt{2\pi}}\int_{-\infty}^\infty (\vA(z)-\vA)\cdot  f_\sigma(z) e^{i\omega z} \rd z , \quad \\ \text{ with } z\ge 0:&\quad \vA(z)-\vA=i\int_0^z [\vH, \vA](s) \rd s
    \end{align}
    Once again by the reverse triangle inequality, we can then lower bound the KMS norm of commutators with $\vO$:
    \begin{align}
        \|[ \hat{\vA}(\omega, t), \vO]\|_{\vrho} \geq&  \|[ \vA(t), \vO]\|_{\vrho}\cdot \hat{f}_\sigma(\omega) \\&- \frac{1}{\sqrt{2\pi}} \int_{0}^\infty f_\sigma(z) \int_0^z    \bigg\|\big[[\vH, \vA](s+t), \vO\big] \bigg\|_{\vrho} \rd s\rd z  -\text{(opposite range)}\label{eq:integrand-rev-tri}
    \end{align}
    Our first course of action is to integrate out the $z$ degree-of-freedom, by explicitly evaluating the integral over $x=z-s$:
    \begin{align}
        &\int_{0}^\infty f_\sigma(z) \int_0^z    \bigg\|\big[[\vH, \vA](s+t), \vO\big] \bigg\|_{\vrho} \rd s\rd z = \\=&  \int_{0}^\infty \bigg\|\big[[\vH, \vA](s+t), \vO\big] \bigg\|_{\vrho} \cdot  \bigg(\int_{0}^\infty f_\sigma(s+x) \rd x\bigg) \rd s  \\
       \leq & \frac{1}{\sqrt{\sigma}} \cdot  \int_{0}^\infty \bigg\|\big[[\vH, \vA](s+t), \vO\big] \bigg\|_{\vrho} \cdot e^{-\sigma^2 s^2} \rd s.
    \end{align}
Analogous to the proof of \cref{gt-dirichlet}, our next goal is to express the norm of the double commutator $[[\vH, \vA](u), \vO]$ in terms of the $g-$Dirichlet form. This is immediate from the calculations in \eqref{eq:comm-comm-A-gt}, \eqref{eq:multi-to-single-gt}:
\begin{align}
    \bigg\|\big[[\vH, \vA](s+t), \vO\big] \bigg\|_{\vrho} \leq \eta_2\bigg(\alpha\sum_{\vB\in \CS_i^1}   \|[\vB(s+t), \vO]\|_{\vrho} + \sJ\cdot \sum_{j\sim i} \sum_{\vB\in \CS_j^1}  \|[\vB(s+t), \vO]\|_{\vrho}\bigg)
\end{align}
with $\eta_2:= (2^{2\sq}c_\beta)^\sk.$ To proceed, we will require the following bounds on integrals over $s, t$, $\forall \vB:$
\begin{align}
      &\iint_{-\infty}^\infty \bigg\|\big[\vB(s+t), \vO\big] \bigg\|_{\vrho}^2 \cdot e^{-\sigma^2 s^2}  g(t)\cdot \rd t\rd s \\ =&\int_{-\infty}^\infty \bigg\|\big[\vB(u), \vO\big] \bigg\|_{\vrho}^2 \cdot \bigg|\int_{-\infty}^\infty e^{-\sigma^2 s^2}  g(u-s)\cdot \rd s\bigg| \rd u
      \\ \leq &\frac{\sqrt{\pi}e^{\pi^2/(\sigma^2\beta^2)}}{\sigma} \int_{-\infty}^\infty \bigg\|\big[\vB(u), \vO\big] \bigg\|_{\vrho}^2 g(u) \rd u\\ =& \frac{\sqrt{\pi}e^{\pi^2/(\sigma^2\beta^2)}}{\sigma} \CC_\vB[\vO]  .
    \end{align}
    which exposes the $g$-Dirichlet form. In the above, we leveraged the explicit formula for $g(x)=\frac{1}{\beta\cosh (2\pi x/\beta)}$. 
    We can now integrate \eqref{eq:integrand-rev-tri} over $t$ under the filter function $g(t)$. We focus on the error term:
    \begin{align}
       &\iint_{-\infty}^\infty   g(t)\cdot   \|[ \vA(t), \vO]\|_{\vrho}\cdot \bigg\|\big[\vB(s+t), \vO\big] \bigg\|_{\vrho} \cdot e^{-\sigma^2 s^2} \rd s \rd t \\ \leq& \CC_{\vA}[\vO]^{1/2} \cdot \bigg| \iint_{-\infty}^\infty \bigg\|\big[\vB(s+t), \vO\big] \bigg\|_{\vrho}^2 \cdot e^{-\sigma^2 s^2}  g(t)\cdot \rd t\rd s\bigg|^{1/2} \cdot \bigg|\int \rd s e^{-\sigma^2 s^2}\bigg|^{1/2} \\
       \leq & \CC_{\vA}[\vO]^{1/2}\cdot \CC_{\vB}[\vO]^{1/2}\cdot \frac{\sqrt{\pi}e^{\pi^2/(2\sigma^2\beta^2)}}{\sigma }\label{eq:CC-exp}
       \end{align}
       To proceed we require two tail bounds on Gaussian integrals over the frequency $\omega$. Here, we explicitly assume $\sigma \leq \beta^{-1}$, such that:
    \begin{equation}
        \int h_\sM(\omega)\cdot \hat{f}_\sigma(\omega)^2 \rd \omega \geq \frac{1}{2}, \quad  \int  h_\sM(\omega)\cdot \hat{f}_\sigma(\omega) \rd \omega \leq 4 \sqrt{\sigma}.
    \end{equation}
    Now integrating \eqref{eq:CC-exp} over $\omega$ and summing over single-site jump operators $\vA^a\in \CS_{i}^1$ gives:
       \begin{align}
           \CE[\vO]&\geq \frac{1}{2}\CC[\vO]-\frac{\eta_2\cdot e^{\pi^2/(2\sigma^2\beta^2)}}{\sigma}\cdot \sum_i\CC_i[\vO]^{1/2}\cdot \bigg(\alpha \CC_i[\vO]^{1/2}+\sJ\sum_{j\sim i} \CC_j[\vO]^{1/2}\bigg)\\
           &\geq \bigg[\frac{1}{2}-\frac{\eta_2\cdot e^{\pi^2/(2\sigma^2\beta^2)}}{\sigma}\cdot \sJ \sd\bigg]\cdot \CC[\vO] \\
           &\geq \bigg[\frac{1}{2} - \beta\sJ\sd\cdot (c_1^\sq c_\beta)^{\sk}\bigg]\cdot \CC[\vO]
       \end{align}
Where in the last line we picked $\sigma=\beta^{-1}$. To conclude the proof, we pick $\beta\sJ\sd$ sufficiently small such that \cref{lemma:conv-cluster} implies the convergence of complex-time evolution with $c_\beta=1+\mathsf{conv}_\beta\leq 2$. Then, there exists a universal constant $\mathsf{c}>1$ such that $\beta\sJ\sd \leq \mathsf{c}^{-\sq\sk}$ implies 
\begin{align}
    \CE[\vO]&\geq \frac{1}{4}\CC[\vO]
\end{align}
as desired.
\end{proof}

\printbibliography

\end{document}